\documentclass[10pt,journal,compsoc]{IEEEtran}

\ifCLASSOPTIONcompsoc
  \usepackage[nocompress]{cite}
\else
  \usepackage{cite}
\fi

\ifCLASSINFOpdf
\usepackage[pdftex]{graphicx}
\usepackage{subfigure}
\usepackage{amsthm}
\usepackage{amssymb}
\usepackage{amsmath}
\usepackage{algorithmic}
\usepackage{algorithm}
\usepackage{enumerate}
\usepackage{stfloats}
\usepackage[justification=centering]{caption} 
\newtheorem{thm}{\bf Theorem}
\usepackage{color,soul}

\soulregister\cite7 
\soulregister\citep7 
\soulregister\citet7 
\soulregister\ref7 
\soulregister\pageref7 
\usepackage{makecell}
\usepackage{float}

\else

\fi

\begin{document}
\title{How to Share: Balancing Layer and Chain Sharing in Industrial Microservice Deployment}

\author{Yuxiang~Liu,
        Bo~Yang,~\IEEEmembership{Senior Member,~IEEE,}
        Yu~Wu,~\IEEEmembership{Student Member,~IEEE,}
        Cailian~Chen,~\IEEEmembership{Member,~IEEE,}
        and~Xinping~Guan,~\IEEEmembership{Fellow,~IEEE}
\thanks{ Y. Liu, B. Yang (Corresponding author), Y. Wu, C. Chen, and X. Guan are with the Department of Automation, Shanghai Jiao Tong University, Shanghai 200240, China; Key Laboratory of System Control and Information Processing, Ministry of Education of China, Shanghai 200240, China; Shanghai Engineering Research Center of Intelligent Control and Management, Shanghai 200240, China (e-mail: liu953973860@sjtu.edu.cn; bo.yang@sjtu.edu.cn; 5wuuy5@sjtu.edu.cn; cailianchen@sjtu.edu.cn; xpguan@sjtu.edu.cn).}
}

\markboth{Journal of \LaTeX\ Class Files,~Vol.~14, No.~8, August~2015}%
{Liu \MakeLowercase{\textit{et al.}}: Bare Demo of IEEEtran.cls for IEEE Journals}

\maketitle

\begin{abstract}
  With the rapid development of smart manufacturing, edge computing-oriented microservice platforms are emerging as an important part of production control. In the containerized deployment of microservices, layer sharing can reduce the huge bandwidth consumption caused by image pulling, and chain sharing can reduce communication overhead caused by communication between microservices. The two sharing methods use the characteristics of each microservice to share resources during deployment. However, due to the limited resources of edge servers, it is difficult to meet the optimization goals of the two methods at the same time. Therefore, it is of critical importance to realize the improvement of service response efficiency by balancing the two sharing methods. This paper studies the optimal microservice deployment strategy that can balance layer sharing and chain sharing of microservices. We build a problem that minimizes microservice image pull delay and communication overhead and transform the problem into a linearly constrained integer quadratic programming problem through model reconstruction. A deployment strategy is obtained through the successive convex approximation (SCA) method. Experimental results show that the proposed deployment strategy can balance the two resource sharing methods. When the two sharing methods are equally considered, the average image pull delay can be reduced to 65\% of the baseline, and the average communication overhead can be reduced to 30\% of the baseline.
\end{abstract}


\IEEEpeerreviewmaketitle

\section{Introduction}
\IEEEPARstart{W}{ith} the rapid development of smart manufacturing and flexible production, the flexibility of industrial production has been greatly enhanced \cite{wangLowlatencyInteroperableIndustrial2020, chenResourceRecommendationModel2022}. Industrial software needs to quickly redistribute and adjust production processes according to changes of orders, which has higher requirements for flexibility and scalability of industrial software \cite{mendoncaDevelopingSelfAdaptiveMicroservice2021, nieHypergraphicalRealtimeMultiRobot2021, hazraCooperativeTransmissionScheduling2022}. Traditional industrial software adopts a monolithic service architecture. The high coupling and occupancy rate within the service will increase the complexity of the whole system. Its scalability, stability, and fault tolerance are difficult to meet the requirements of smart manufacturing. Therefore, the industrial software architecture based on microservices has been widely concerned \cite{thramboulidisCyberphysicalMicroservicesIoTbased2018, yangMIRASModelbasedReinforcement2019}. Through the microservice architecture, a complete service can be split into multiple loosely coupled microservices. Different microservices are logically independent and have a high degree of flexibility, scalability, and fault tolerance, which can well adapt to the requirements of smart manufacturing.


To meet the high requirements of computation-intensive tasks for real-time performance and service efficiency in smart manufacturing, edge computing-oriented microservice platforms are emerging \cite{tianDIMADistributedCooperative2021, baoPerformanceModelingWorkflow2019,jianCloudEdgebasedTwolevel2021,shiEdgeComputingVision2016,zhangFengHuoLunFederatedLearning2020}. At present, container technologies represented by Docker \cite{merkelDockerLightweightLinux2014} and container orchestration tools represented by Kubernetes \cite{burnsBorgOmegaKubernetes2016} are becoming mainstream solutions for microservice deployment and maintenance on edge platforms. According to different service requests and deployment strategies, each microservice which is packaged into a Docker image can be deployed to edge servers through container orchestration tools.


In the containerized deployment of microservices, service efficiency is an important indicator for evaluating the quality of the deployment solution. Service efficiency is mainly affected by two aspects. One is the startup time of microservices. It mainly depends on the pull delay of Docker images which are stored in the cloud through different image layers \cite{guExploringLayeredContainer2021}. When a microservice needs to be provided locally, the edge server will pull a non-local container image containing all required layers from the cloud. Due to limited network bandwidth, image pulls incur a corresponding downlink delay. A comprehensive research shows that with a bandwidth of 100Mbps, the average startup time of a single image is about 20.7 seconds, while the average image pull delay is about 15.8 seconds, accounting for 76.6\% of the average startup time \cite{guNDockerNVMHDDHybrid2019}. Image pull delay has become a non-negligible factor affecting container startup time, which in turn affects the efficiency of service response. The other is the communication overhead between microservices. It depends on the amount of data communicated between microservices. An industrial application can be completed by multiple microservices deployed on one or more edge servers \cite{wangMPCSMMicroservicePlacement2021}. These microservices can be called microservice chains, and there will be frequent data exchanges between microservices in the same microservice chain \cite{lvMicroserviceDeploymentEdge2022}. A large amount of data transmission between microservices will cause high transmission delay, which will affect the service response efficiency.

Due to the above two aspects, it is very important to improve service efficiency through resource sharing. There are two types of resource sharing strategies for the improvement of service efficiency. One of the strategies for resource sharing is layer sharing \cite{guExploringLayeredContainer2021}. Docker natively supports the sharing of layers. If the microservices deployed on the same edge server use the same image layer, the layer will not be pulled repeatedly when pulling images. This layer can be shared by all microservices on the server. The image pull delay can be effectively reduced by layer sharing, thereby improving the startup speed and service response efficiency of microservices. The other strategy for resource sharing is chain sharing \cite{lvMicroserviceDeploymentEdge2022, yuJointOptimizationService2019}, which can be defined as the data sharing of microservices deployed on the same server. In the microservice chain, there is frequent data transfer between two adjacent microservices. If two microservices are deployed on the same server, the data can be directly accessed through chain sharing by the next microservice without multi-hop transmission of data. The delay and packet loss caused by data transmission can be reduced by chain sharing.

However, due to the limited resources of edge servers, it is impossible for all microservices to be deployed on the same edge server. Therefore, it is necessary to find an optimal microservice deployment strategy for the trade-off between layer sharing and chain sharing. Besides service efficiency, due to the limited resources of edge servers, the microservice deployment strategy can not make full use of different resources (such as computing and storage resources) at the same time, resulting in idle computing resources. Therefore, a method is also needed to reasonably allocate resources to different microservices deployed on a server and maximize the utilization of resources.


Aiming at resource sharing and maximizing resource utilization problems among microservices, the deployment of microservices mainly faces the following difficulties. 1) How to model the layered structure of the microservice image to accurately describe the relationship between the microservice image and the container layer. 2) How to describe the chain structure of microservices and the communication between microservices. 3) How to balance layer sharing and chain sharing to establish an optimization problem to achieve the best deployment strategy. 4) How to reallocate resources to microservices deployed on edge servers to make full use of computing resources. In this paper, we study the microservice deployment problem considering microservice layer sharing and chain sharing. The problem is modeled as an integer programming problem that minimizes image pull delay and communication overhead. Based on this problem, a microservice deployment strategy and resource redistribution scheme are proposed. The main contributions of this work are as follows: 

\begin{enumerate}[1)]
  \item We describe the layered structure and chain structure of microservices through the same model. An integer programming problem is established to minimize the image pull delay and communication overhead.
  \item Through model reconstruction, we prove that the integer programming problem can be transformed into an integer quadratic programming problem with linear constraints. The optimal solution is obtained by using the successive convex approximation (SCA) method. This method can effectively balance the image pull delay and communication overhead.
  \item A resource redistribution algorithm for edge servers is proposed to make full use of idle computing resources.
  \item Through experiments, the results are evaluated in multiple dimensions, such as image pull delay and inter-service communication overhead. These experiments demonstrate the effectiveness of the proposed method.
\end{enumerate}

The remainder of this paper is organized as follows. Sec. \ref{Related Works} briefly reviews the related literature. In Sec. \ref{Modeling}, the layered structure and chain structure of the system are modeled, and the problem formulation is given. Sec. \ref{Solution1} solves the proposed problem. Sec. \ref{reallocation} proposes a resource redistribution algorithm for edge servers. Sec. \ref{performance} evaluates the results of the proposed method. Sec. \ref{discussion} discusses the limitations and future work. Sec. \ref{conclusion} concludes the paper.

\section{Related Works} \label{Related Works}

In this section, we discuss current research on the deployment of microservices.

In recent years, optimizing the cost and improving microservice response efficiency have received wide attention. Herrera \emph{et al.} \cite{herreraOptimizingResponseTime2021a} designed a distributed microservice deployment framework named DADO to optimize the response time of microservices. Deng \emph{et al.} \cite{dengOptimalApplicationDeployment2021} and Chen \emph{et al.} \cite{chenDynamicServiceMigration2022} proposed algorithms to solve the cost-aware microservice deployment problem, they considered application deployment cost and service migration cost, respectively. Fadda \emph{et al.} \cite{faddaMonitoringAwareOptimalDeployment2021} provided an approach for supporting the deployment of microservices in multi-cloud environments to optimize the quality and cost. Zhao \emph{et al.} \cite{zhaoCostEfficientEdgeIntelligent2020a} developed a cost-aware elastic microservice deployment algorithm to solve the container-based microservice deployment problem. The above researchers have conducted sufficient research on service response efficiency and service quality. However, these studies do not consider the characteristics of microservices, such as the chain structure of multiple microservices and the layered structure due to containerized deployment.


From the perspective of the chain structure, deployment strategies become more complex due to the dependencies between microservices. In general, a microservice chain can be modeled as a directed acyclic graph \cite{faticantiThroughputAwarePartitioningPlacement2020, josephIntMADynamicInteractionaware2020}. Wang \emph{et al.} \cite{wangMPCSMMicroservicePlacement2021} and Li \emph{et al.} \cite{liOnlineReconfigurationLatencyAware2021} proposed algorithms to solve latency-aware microservice deployment problems. Armani \emph{et al.} \cite{armaniCostEffectiveWorkloadAllocation2021} proposed a cost-effective workload distribution strategy for microservice-based applications considering fault tolerance and load balancing of microservice chains. Sasabe \emph{et al.} \cite{sasabeCapacitatedShortestPath2021} considered the service chaining and function placement problems to optimize the total delay. Lv \emph{et al.} \cite{lvMicroserviceDeploymentEdge2022} considered the containerized deployment of microservices, and a chain sharing deployment strategy is proposed to minimize the communication overhead. The above research focuses on the chain structure and chain sharing of microservices, but it does not take into account the layered structure in containerized deployments of microservices.


\begin{figure*}[t]
  \centering
  \includegraphics[width=5.5in]{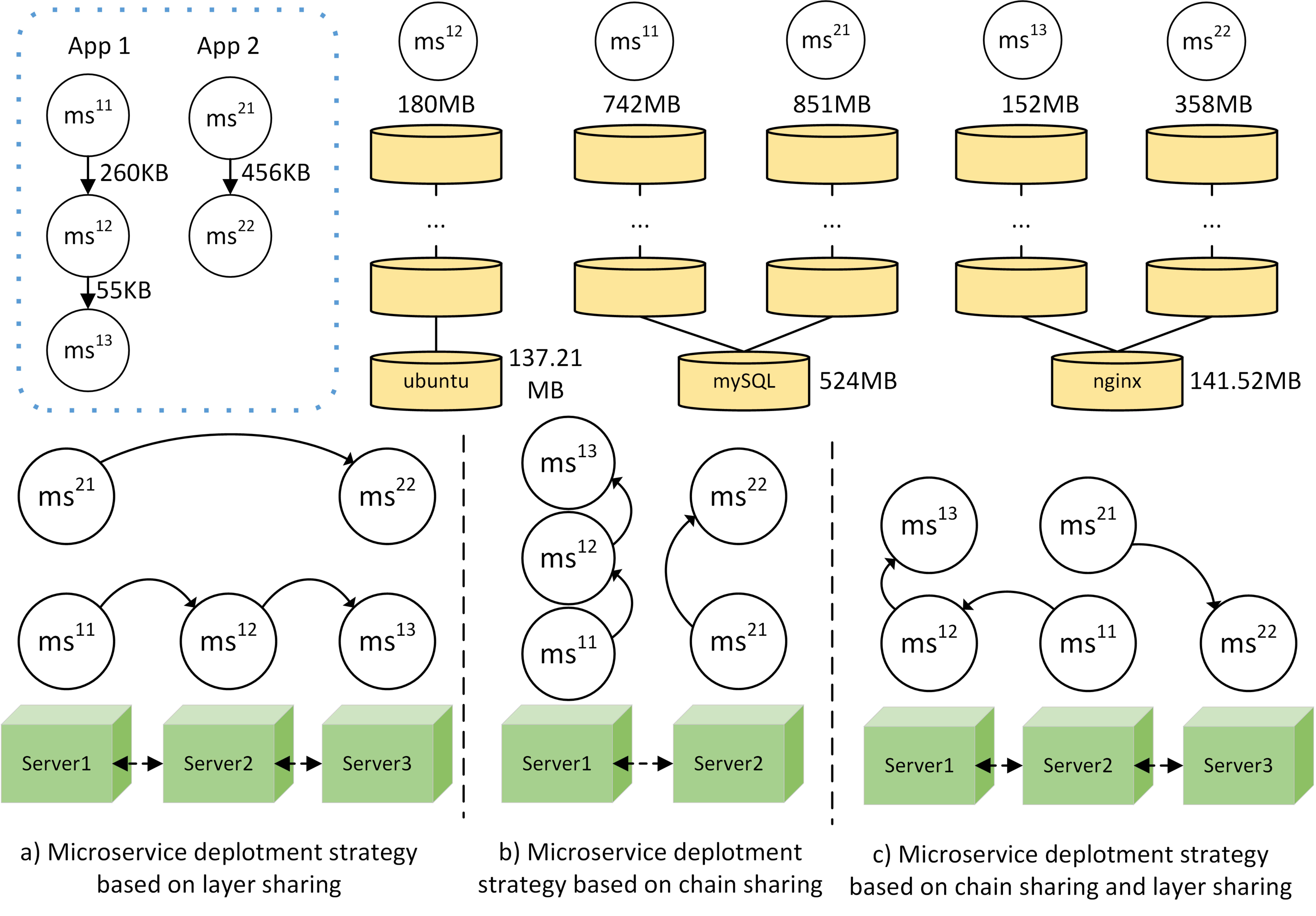}
  \caption{An example of layer sharing and chain sharing deployment strategies}
  \label{model}
\end{figure*}

For the layered structure, researchers have focused on how to reduce the latency of image pulling by reducing image size or utilizing layer sharing of images. Lou \emph{et al.} \cite{louEfficientContainerAssignment2022} considered layer sharing among images and proposed a layer-aware scheduling algorithm. Gu \emph{et al.} \cite{guLayerAwareMicroservice2021} designed a microservice deployment and request scheduling strategy based on layer sharing and used an iterative greedy algorithm to obtain the optimal strategy to improve the throughput of microservices. Gu \emph{et al.} \cite{guLayerawareCollaborativeMicroservice2022} also investigated the problem of how to collaboratively deploy microservices by incorporating both intra-server and inter-server layer sharing to maximize the edge throughput. The above research fully considers the layer sharing strategy in the containerized deployment of microservices, but does not consider the sharing strategy in the presence of the chain structure between microservices.

Although the existing schemes have considered optimizing the efficiency and cost of microservice deployment in terms of the layered structure and chain structure, respectively, there are still some challenging problems to be solved. First, in complex intelligent manufacturing scenarios, the layered structure and chain structure usually coexist. In this case, the deployment strategies of layer sharing and chain sharing will affect each other, which further increases the difficulty of finding an effective deployment strategy. Secondly, an uneven deployment strategy will lead to idle server resources due to inconsistent server resources required by different microservices. In order to solve these problems, we propose a microservice deployment scheme that comprehensively considers layer sharing and chain sharing. It can both reduce the delay of image pulling and the communication overhead. The idle server resources can also be fully utilized. Our scheme can effectively improve the operating efficiency of microservices. To the best of our knowledge, there is no research that comprehensively considers the two sharing methods.

\section{System Modeling and Problem Formulation} \label{Modeling}

\subsection{A simple example}

First, we show the two strategies of layer sharing and chain sharing through a simple microservice deployment model. As shown in Fig. \ref{model}, we consider two applications composed of two and three microservices, respectively. We denote the $i$th microservice in application $k$ as $ms^{ki}$. Each microservice image consists of a different number of image layers. The bandwidth between the three servers and the cloud server is 120 MB/s, and the two adjacent servers can be reached with a single hop. Under the layer sharing deployment strategy, the same image layer on the same edge server can be shared. Therefore, the size of the image layer to be pulled is 1617.48 MB, and the total download time is 13.479 seconds. However, the total communication data is 1031 KB because of communication between servers. Under the chain sharing deployment strategy, the total communication data is 0 KB since the microservices in the same chain are all deployed on the same server. The size of the image layer to be pulled is 2283 MB because there is no layer sharing, and the total pull delay is 19.025 seconds. When considering both chain sharing and layer sharing, we can get the result shown in Fig. \ref{model}(c). The size of the image layer to be pulled is 1758 MB, the total download time is 14.65 seconds, and the total communication data is 315 KB. These data can be found in Table \ref{example}. It can be seen that different deployment strategies have a significant impact on image pull delay and communication overhead. If layer sharing and chain sharing can be both considered, we will get low image pull delay and low communication overhead at the same time.

\begin{table}[tbp]
	\centering 
	\caption{Comparison about layer sharing and chain sharing}
	\label{example}
	\begin{tabular}{p{52pt}|p{65pt}|p{95pt}} \hline \hline
    \rule{0pt}{8pt} \quad & \textbf{image pull delay} & \textbf{communication overhead}\\ \hline
    \rule{0pt}{8pt} \textbf{layer sharing} & 13.479s & 1031 KB \\ \hline
    \rule{0pt}{8pt} \textbf{chain sharing} & 19.025s & 0 KB \\ \hline
    \rule{0pt}{8pt} \textbf{both} & 14.65s & 315 KB \\ \hline
    \hline
	\end{tabular}
\end{table}

\subsection{System model}

We consider an intelligent manufacturing system, which has $M$ production devices and $N$ edge servers. We define device set as $\mathbb{M}=\{1,2,\cdots,M\}$ and server set as $\mathbb{N}=\ \{1,2,\cdots,N\}$. A cloud server is deployed at the remote end to store the microservice images. Each production device is connected to the nearest edge server. Each edge server has limited computing and storage resources, and a certain number of microservices can be deployed on it. The computing and storage resources of edge server $n$ are denoted as $C_n^C$ and $C_n^S$, and the bandwidth between cloud server and edge server $n$ is $b_n^{cloud}$.

Suppose there are several industrial applications, and application set is defined as $\mathbb{K}=\{1,2,\cdots,K\}$. Each application is composed of multiple microservices. The microservice set in the $k$th application is $\mathbb{A}_k=\{1,2,\cdots,A_k\}$, where $A_k$ is the amount of microservices in the $k$th application. Each application can handle service requests from production device. We use $ms^{ki}$ to denote the $i$th microservice in application $k$, and $u^{ki}$ to denote the computing resources requested by microservice $ms^{ki}$.

All microservice images are stored in the microservice image library of the cloud server, and are pulled by the edge server according to the deployed microservices. Every microservice image consists of some shareable layers and some non-shareable layers. We use set $\mathbb{L}=\{1,2,\cdots,L\}$ to represent all layers of different size and $S^l\in\ \mathbb{R}^+$ to represent the size of layer $l\ \in\ \mathbb{L}$. In this way, each microservice can be composed of one or more layers in $\mathbb{L}$, and $E^{kil}\in\{0,1\}$ can be used to indicate whether $ms^{ki}$ contains the $l$th layer. $E^{kil}=1$ represents that $ms^{ki}$ contains the $l$th layer.

Each server will receive different service requests and data. If the expected microservice is deployed on the server at this time, the service request can be processed directly. If the expected microservice is not deployed on the server, the request and data need to be transmitted to another edge server through multi-hop transmission. Due to the different geographical locations, the hops in communication between different servers is also different. We define $D_{nn^\prime}$ as the hops of requests or data transmitted from server $n$ to server $n^\prime$, which can be obtained by the shortest communication path between the two servers. It is obvious that $D_{nn^\prime} = D_{n^\prime n}$, $D_{nn}=0$. We can use a matrix $\mathbf{D}$ to represent the multi-hop connection between all servers. The notations and variables commonly used in this paper are summarized in Table \ref{notation}.

\begin{equation}
  \mathbf{D}=\begin{bmatrix}0&D_{12}&D_{13}&\cdots&D_{1N}\\D_{21}&0&D_{23}&\cdots&D_{2N}\\D_{31}&D_{32}&0&\cdots&D_{3N}\\\vdots&\vdots&\vdots&\ddots&\vdots\\D_{N1}&D_{N2}&D_{N3}&\cdots&0\\\end{bmatrix}\nonumber
\end{equation}

\begin{table}[tbp]
	\centering 
	\caption{Commonly used notations and variables}
	\label{notation}
	\begin{tabular}{p{30pt}|p{180pt}} \hline \hline
    \rule{0pt}{8pt}\textbf{Symbol} & \textbf{Description}\\ \hline
    \rule{0pt}{8pt} $M,\mathbb{M}$ & devices and the set of devices \\ \hline
    \rule{0pt}{8pt} $N,\mathbb{N}$ & servers and the set of servers \\ \hline
    \rule{0pt}{8pt} $K,\mathbb{K}$ & applications and the set of applications \\ \hline
    \rule{0pt}{8pt} $L,\mathbb{L}$ & layers and the set of layers \\ \hline
    \rule{0pt}{8pt} $A_k,\mathbb{A}_k$ &microservices in the $k$th application and the set of microservices \\ \hline
    \rule{0pt}{8pt} $C_n^C$ &the computing resources of edge server $n$ \\ \hline
    \rule{0pt}{8pt} $C_n^S$ &the storage resources of edge server $n$ \\ \hline
    \rule{0pt}{8pt} $b_n^{cloud}$ &the bandwidth between cloud server and edge server $n$ \\ \hline
    \rule{0pt}{8pt} $ms^{ki}$ &the $i$th microservice in application $k$ \\ \hline
    \rule{0pt}{8pt} $u^{ki}$ &the computing resources requested by microservice $ms^{ki}$ \\ \hline
    \rule{0pt}{8pt} $S^l$ &the size of layer $l$ \\ \hline
    \rule{0pt}{8pt} $E^{kil}$ &whether $ms^{ki}$ contains the $l$th layer \\ \hline
    \rule{0pt}{8pt} $D_{nn^\prime}$ &the hops of requests or data transmitted from server $n$ to server $n^\prime$ \\ \hline
    \rule{0pt}{8pt} $x_n^{ki}$ &whether $ms^{ki}$ is deployed in server $n$ \\ \hline
    \rule{0pt}{8pt} $d_n^l$ &whether the layer $l$ is pulled to edge server $n$ \\ \hline
    \hline 
	\end{tabular}
\end{table}
\subsection{Problem Formulation}

\subsubsection{Microservice deployment and layer sharing}

We define $x_n^{ki}\in\{0,1\}$ to represent the deployment of $ms^{ki}$, and $x_n^{ki} = 1$ to represent that the microservice is deployed on the edge server $n$, otherwise not. Due to the layered structure of microservices, once a microservice is deployed on edge server $n$, all layers contained in the microservice image need to be pulled to server $n$. We use the variable $d_n^l\in\{0,1\}$ to represent whether the layer $l$ is pulled to edge server $n$, and $d_n^l=1$ indicates that the $l$th layer needs to be downloaded to the edge server $n$, otherwise not.

Since each microservice can only be deployed on a unique server, we can get the following constraints:

\begin{equation}
  \sum_{n\in\mathbb{N}} x_n^{ki}=1,\forall k\in\mathbb{K},\forall i\in\mathbb{A}_k \label{s1}
\end{equation}

If microservices deployed on the same server can share the same layer, the layer only needs to be downloaded once. Therefore, $d_n^l$ and $x_n^{ki}$ satisfy the following constraints:

\begin{equation}
  d_n^l=\min{\{\sum_{k\in\mathbb{K}}\sum_{i\in\mathbb{A}_k} x_n^{ki}}E^{kil},1\},\forall l\in\mathbb{L},\forall n\in\mathbb{N} \label{s2}
\end{equation}

Due to the limited storage resources, the layer size of the deployed microservices needs to be smaller than the storage resources of edge servers. So we can get the following constraints:

\begin{equation}
  \sum_{l\in\mathbb{L}} d_n^lS^l\le C_n^S,\forall n\in\mathbb{N} \label{s3}
\end{equation}

Due to the limited computing resources, the total computing resource of all microservices deployed on a server needs to be less than the computing resource of the server. So we can get the following constraints:

\begin{equation}
  \sum_{k\in\mathbb{K}}\sum_{i\in\mathbb{A}_k} x_n^{ki}u^{ki}\le C_n^C,\forall n\in\mathbb{N} \label{s4}
\end{equation}

For each server, all layers deployed on the server need to be pulled from the cloud. The image pull delay of server $n$ can be expressed as follows:

\begin{equation}
  T_n=\frac{\sum_{l\in \mathbb{L}} d_n^lS^l}{b_n^{cloud}}
\end{equation}

\subsubsection{Communication overhead and chain sharing}

We model the microservice chain as a directed weighted acyclic graph to reveal the impact of communication data on the deployment of microservices. Taking an application consisting of four microservices as an example, the modeled directed weighted graph is shown in Fig. \ref{chain}. We use the interaction weight $w_{ij}^k$ to represent the size of the communication traffic between each two microservices $ms^{ki}$ and $ms^{kj}$.

\begin{figure}[htbp]
  \centering
  \includegraphics[width=2in]{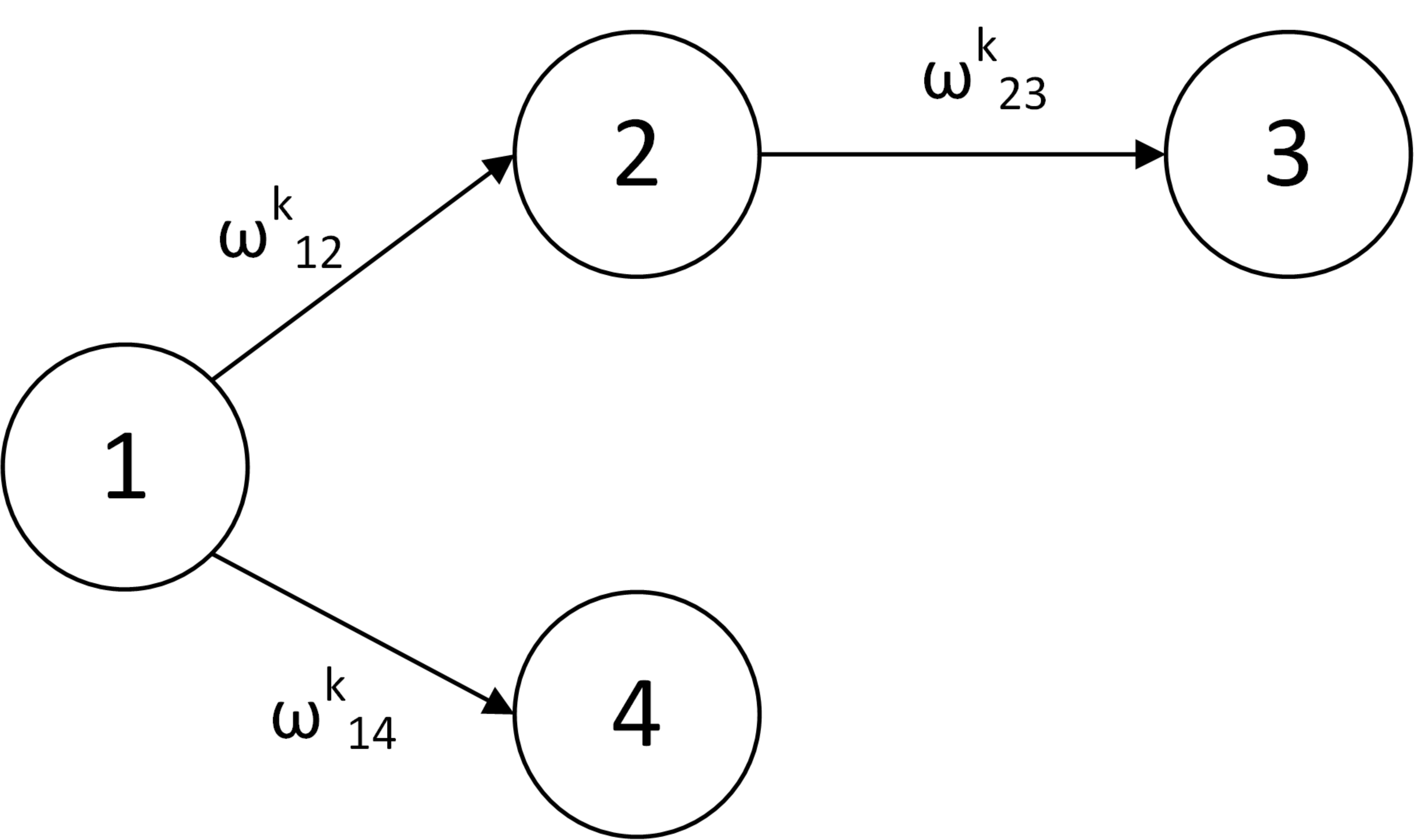}
  \caption{Directed weighted graph for application $k$}
  \label{chain}
\end{figure}

The interaction graph can be written in the form of a matrix. For application $k$, its interaction matrix is defined as follows:

$$
\mathbf{w}^k=\begin{bmatrix}w_{11}^k&\cdots&w_{1A_k}^k\\\vdots&\ddots&\vdots\\w_{A_k1}^k&\cdots&w_{A_k,A_k}^k\\\end{bmatrix}
$$
where $w_{ij}^k$ is non-zero value only when $ms^k_i$ and $ms^k_j$ are connected. For example, $\mathbf{w}^k$ of the microservice chain shown in Fig. \ref{chain} can be defined as follows:

$$
\mathbf{w}^k=\begin{bmatrix}0&w_{12}^k&0&w_{14}^k\\0&0&w_{23}^k&0\\0&0&0&0\\0&0&0&0\\\end{bmatrix}
$$

Based on the interaction graph, we can calculate the communication cost. Multi-hop data transmission between two adjacent microservices is not required if they are deployed on the same edge server. To calculate the total amount of data transferred between microservices in an application, we first need to find the hops between the servers where any two microservices are deployed. For the servers where any two microservices $ms^{ki}$ and $ms^{kj}$ are deployed, we define $Hop\left(k,i,j\right)=\sum_{n\in\mathbb{N}}\sum_{n^\prime\in\mathbb{N}} x_n^{ki}x_{ n^\prime}^{kj}D_{nn^\prime}$ to calculate the hops. For any application $k$, its communication overhead can be expressed as follows:

\begin{align}
  R^k&=\sum_{i\in\mathbb{A}_k}\sum_{j\in\mathbb{A}_k}{w_{ij}^kHop\left(k,i,j\right)} \nonumber \\
  &= \sum_{i\in\mathbb{A}_k}\sum_{j\in\mathbb{A}_k}\left(w_{ij}^k\sum_{n\in N}\sum_{n^\prime\in N} x_n^{ki}x_{n^\prime}^{kj}D_{nn^\prime}\right)
\end{align}

\subsubsection{The virtual microservice}

Each application $k$ originates from a service request on a production device. Each generated service request is transmitted via the network to the edge server closest to the production device at first. We define the number of the source device for the application $k$ as $source^k$. For each $source^k$, we find its directly connected edge server $N^k$ and define a virtual initial microservice $ms^{k0}$ to describe the impact of request generation location on microservice deployment. We use $ms^{k0}$ to denote the service requests generated on device $source^k$. The microservice set in the $k$th application is modified as $\mathbb{A}_k=\{0,1,2,\cdots,A_k\}$. Therefore, we should add the virtual initial microservice to the  interaction diagram. Fig. \ref{chain} can be modified as follows:

\begin{figure}[htbp]
  \centering
  \includegraphics[width=2.5in]{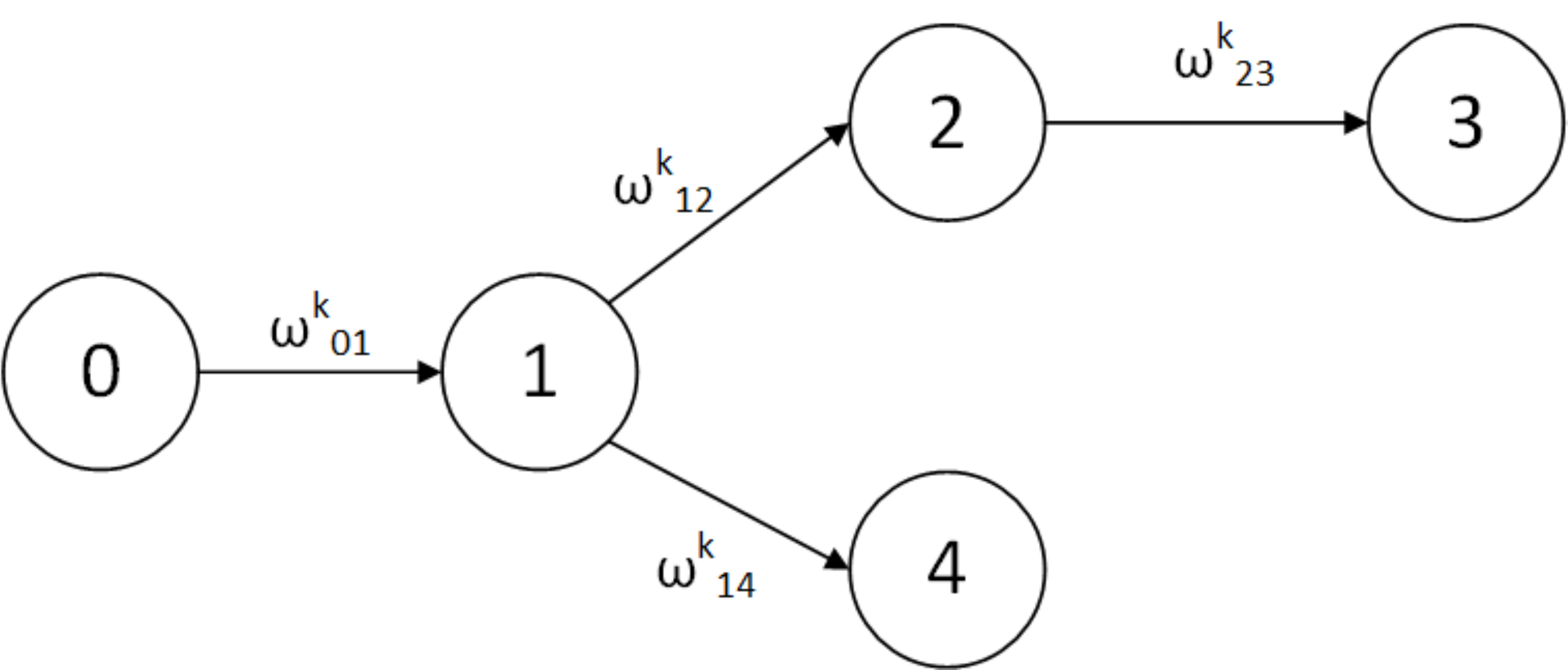}
  \caption{Modified directed weighted graph for application $k$}
  \label{chain_ini}
\end{figure}

The microservice $ms^{k0}$ does not actually exist, so its required computing resource is $u^{k0} = 0$ and does not contain any layers. When its deployment location is fixed, we can get the following constraints:

\begin{equation}
  x_{N^k}^{k0}=1,\forall k\in\mathbb{K} \label{s5}
\end{equation}

\subsubsection{Image pull delay and communication overhead minimization problem}

The goal of microservice deployment is to minimize image pull delay and communication overhead under the constraints of device resources and service characteristics. The optimization problem can be expressed as follows:

\begin{align}
  \text{P1:} \quad \min_{x,d} \quad & T,R \\
  \text{s.t.} \quad &(\ref{s1}),(\ref{s2}),(\ref{s3}),(\ref{s4}),(\ref{s5}) \nonumber \\
  &x_n^{ki},d_n^l\in\{0,1\} \label{s6}
\end{align}
where $T=\sum_{n\in\mathbb{N}} T_n$ is the total image pull delay. $R = \sum_{k\in\mathbb{K}} R^k$ is the total communication overhead. This problem is a multi-objective optimization problem and there is a multiplicative form of variables in $R^k$. Therefore, the problem is difficult to solve. In next section, we will transform the problem to a single-objective optimization problem and give a solution.

\section{Microservice Deployment Scheme based on SCA} \label{Solution1}

\subsection{Problem transformation} 

We vectorize all variables through model reconstruction to make the problem clearer and easier to solve. Then we convert all constraints to linear constraints. Finally, the problem is transformed into a single-objective integer quadratic programming problem through an additive weighted model.

\subsubsection{Image pull delay} \label{transform}

Consider the first part of problem $\text{P1}$:

\begin{equation}
  T=\sum_{n\in\mathbb{N}}\frac{\sum_{l\in L} d_n^lS^l}{b_n^{cloud}}
\end{equation}
which is a linear form. We define $\mathbf{d}=\left[\mathbf{d}_1^T,\cdots,\mathbf{d}_N^T\right]^T$, $\mathbf{S}=\left[S^1,\cdots,S^L\right]^T$, and $\mathbf{M}=\left[\frac{\mathbf{S}^T}{b_1^{cloud}},\cdots,\frac{\mathbf{S}^T}{b_N^{cloud}}\right]$. Then the calculation of the total pull delay can be converted to

\begin{equation}
  T=\mathbf{Md}
\end{equation}

\subsubsection{Communication overhead}

Consider the second part of problem $\text{P1}$:

\begin{equation}
  R=\sum_{k\in\mathbb{K}}\sum_{i\in\mathbb{A}_k}\sum_{j\in\mathbb{A}_k}\left(w_{ij}^k\sum_{n\in N}\sum_{n^\prime\in N} x_n^{ki}x_{n^\prime}^{kj}D_{nn^\prime}\right)
\end{equation}

We can also define $\mathbf{x}^{ki}=\left[x_1^{ki},x_2^{ki},\cdots,x_N^{ki}\right]^T$, $\mathbf{x}^k=\left[\left(\mathbf{x}^{k1}\right)^T,\cdots,\left(\mathbf{x}^{kA_k}\right)^T\right]^T$, and $\mathbf{x}=\left[\left(\mathbf{x}^1\right)^T,\cdots,\left(\mathbf{x}^K\right)^T\right]^T$. We can get 

\begin{equation}
  \sum_{n\in N}\sum_{n^\prime\in N} x_n^{ki}x_{n^\prime}^{kj}D_{nn^\prime}=\left(\mathbf{x}^{kj}\right)^T\mathbf{D}\mathbf{x}^{ki}
\end{equation}

Let $\odot$ be the Hadamard product of the matrix and define

\begin{equation}
  \mathbf{W}^k=\mathbf{w}^k\odot\begin{bmatrix}\mathbf{D}&\cdots&\mathbf{D}\\\vdots&\ddots&\vdots\\\mathbf{D}&\cdots&\mathbf{D}\\\end{bmatrix}_{A_k\times A_k}
\end{equation}

\begin{equation}
  \mathbf{W}=\begin{bmatrix}\mathbf{W}^1&\cdots&0\\\vdots&\ddots&\vdots\\0&\cdots&\mathbf{W}^K\\\end{bmatrix}_{K\times K}
\end{equation}

Then the calculation of the communication overhead can be converted to

\begin{equation}
  R=\mathbf{x}^T\mathbf{Wx}
\end{equation}

\subsubsection{Constraints}

Considering equation (\ref{s2}), the original constraint is nonlinear. We can turn it into a linear constraint by two new constraints:

\begin{align}
  d_n^l&\le\sum_{k\in K}\sum_{i\in A_k} x_n^{ki}E^{kil} \label{s2_1} \\
  d_n^l&\geq\frac{\sum_{k\in K}\sum_{i\in A_k} x_n^{ki}E^{kil}}{Z} \label{s2_2}
\end{align}
where $Z$ is an arbitrarily large constant greater than 1. Equations (\ref{s2_1}) and (\ref{s2_2}) can be equivalent to constraint (\ref{s2}) because $d_n^l$ is a binary variable. Therefore, all constraints in problem $P1$ are transformed into linear constraints.

For linear constraints, we can also vectorize all constraints in the same way as in Sec. \ref{transform}, and the transformed constraints are

\begin{align}
  \mathcal{Q}\mathbf{x}&=\mathbf{b}_1 \label{S1}\\ 
  \mathbf{Hx} &= \mathbf{b}_2 \label{S2} \\
  \mathbf{d}&\le\mathbf{Y}\mathbf{x} \label{S3} \\
  \mathbf{d}&\geq\frac{\mathbf{Y}\mathbf{x}}{Z} \label{S4} \\
  \mathcal{S}\mathbf{d}&\le\mathbf{C}^S \label{S5} \\
  \mathcal{G}\mathbf{x}&\le\mathbf{C}^C \label{S6}
\end{align}

Constraints (\ref{S1})-(\ref{S6}) correspond to (\ref{s1}), (\ref{s5}), (\ref{s2_1}), (\ref{s2_2}), (\ref{s3}), (\ref{s4}) respectively. Appendix A shows the detailed values of matrices $\mathcal{Q},\mathbf{b}_1,\mathbf{H},\mathbf{b}_2,\mathbf{Y},\mathcal{S},\mathbf{C}^S,\mathcal{G},\mathbf{C}^C$. Problem $P1$ can be transformed into

\begin{align}
  \text{P2:} \quad \min_{x,d} \quad & T,R\\
  \text{s.t.} \quad &(\ref{s6}), (\ref{S1})-(\ref{S6})  \nonumber 
\end{align}
where $T = \mathbf{Md}$, $R = \mathbf{x}^T\mathbf{Wx}$.

\subsubsection{Single objective optimization problem}

The original problem has two optimization objectives. We use an additive weighting model to turn the original problem into a single-objective problem. The utility function is as follows:

\begin{align}
  F(\mathbf{x,d}) &= \theta\frac{T - T_{min}}{T_{max} - T_{min}} + (1-\theta)\frac{R - R_{min}}{R_{max} - R_{min}}  \nonumber \\
  &= \frac{\theta}{T_{max} - T_{min}}\mathbf{Md} + \frac{1-\theta}{R_{max} - R_{min}}\mathbf{x}^T\mathbf{Wx} \nonumber \\ &\quad + Const
\end{align}
where $Const = -\frac{\theta T_{min}}{T_{max} - T_{min}} - \frac{(1-\theta)R_{min}}{R_{max} - R_{min}}$. $T_{max}$ and $R_{max}$ represent the maximum value of image pull delay and communication overhead, respectively. $T_{min}$ and $R_{min}$ represent the minimum value of image pull delay and communication overhead. $\theta \in [0,1]$ represents the preference for image pull delay and communication overhead. Finally, the original optimization problem can be transformed into a new optimization problem as follows: 

\begin{align}
  \text{P3:} \quad \min_{x,d} \quad &F(\mathbf{x,d}) \\
  \text{s.t.} \quad &(\ref{s6}), (\ref{S1})-(\ref{S6})  \nonumber 
\end{align}

This problem is an integer quadratic programming problem. The solution of $\text{P3}$ is the weakly Pareto optimal solution of the original problem. If $\theta \in (0,1)$, the solution of $\text{P2}$ is the Pareto optimal solution. The proof can be found in \cite{wangDelayAwareMicroserviceCoordination2021}.

\subsection{Solution based on successive convex approximation}

Since $\mathbf{W}$ is not a positive semi-definite matrix, the problem is a non-convex quadratic programming, which is difficult to solve directly. First, we transform $\text{P3}$ into a convex optimization problem. Then the problem can be solved based on SCA \cite{razaviyaynSuccessiveConvexApproximation2014}. Let $\mathbf{Q} = \mathbf{W}+\mathbf{W}^T$, then minimizing $F(\mathbf{x,d})$ is equivalent to minimize

\begin{equation}
U(\mathbf{x,d}) = C_1\mathbf{M}\mathbf{d} + \frac{1}{2}C_2\mathbf{x}^T\mathbf{Qx} \label{sca_1}
\end{equation}
where $C_1 = \frac{\theta}{T_{max} - T_{min}}$, $C_2 = \frac{1-\theta}{R_{max} - R_{min}}$. For the matrix $\mathbf{Q}$, set the eigenvalues of the matrix as $\lambda_1,\lambda_2,\cdots,\lambda_n$. We can define $\lambda_Q = \max \{|\lambda_i|\}$, and the matrix $\mathbf{Q}$ can be split as follows:

\begin{equation}
\mathbf{Q}  = \mathbf{Q} + \lambda_Q\mathbf{I} - \lambda_Q\mathbf{I} = \mathbf{P} - \mathbf{N}
\end{equation}
where $\mathbf{P} = \mathbf{Q} + \lambda_Q\mathbf{I}$, $\mathbf{N} = \lambda_Q\mathbf{I}$. Equation (\ref{sca_1}) becomes

\begin{equation}
  U(\mathbf{x,d}) = U_1(\mathbf{x,d}) - U_2(\mathbf{x})
\end{equation}
where $U_1(\mathbf{x,d}) = C_1\mathbf{M}\mathbf{d} + \frac{1}{2}C_2\mathbf{x}^T\mathbf{Px}$ is convex, and $- U_2(\mathbf{x}) = - \frac{1}{2}C_2\mathbf{x}^T\mathbf{Nx}$ is nonconvex. Next, we need to make a convex approximation to $- U_2(\mathbf{x})$ at $\bar{x}$, where $\bar{x}$ is a point in the feasible set of $\text{P3}$. Let 

\begin{align}
  l(\mathbf{x}) &= -U_2(\bar{\mathbf{x}}) - \nabla U_2(\bar{\mathbf{x}})^T(\mathbf{x} - \bar{\mathbf{x}}) \nonumber \\
  & = - C_2\bar{\mathbf{x}}^T\mathbf{Nx} + \frac{1}{2}C_2\bar{\mathbf{x}}^T\mathbf{N}\bar{\mathbf{x}} \geqslant - U_2(\mathbf{x})
\end{align}

Finally, we get the convex approximation problem

\begin{align}
  \text{P4:} \quad \min & \quad U_{qp}(\mathbf{x,d};\bar{\mathbf{x}},\bar{\mathbf{d}}) \nonumber\\ &= U_1(\mathbf{x,d}) + l(\mathbf{x}) \nonumber\\
  &= C_1\mathbf{Md} + \frac{1}{2}C_2\mathbf{x}^T\mathbf{Px} \nonumber \\ &\quad - C_2\bar{\mathbf{x}}^T\mathbf{Nx} + \frac{1}{2}C_2\bar{\mathbf{x}}^T\mathbf{N}\bar{\mathbf{x}} \\
  \text{s.t.} &\quad (\ref{s6}), (\ref{S1})-(\ref{S6})  \nonumber 
\end{align}

$\text{P4}$ is a convex quadratic programming problem and can be solved directly with the commercial solver. So we can solve $\text{P3}$ by SCA \cite{razaviyaynSuccessiveConvexApproximation2014} algorithm, as shown in Algorithm \ref{alg:SCA}.

\subsection{Convergence analysis}

In this subsection, we show that the SCA algorithm can reach the optimal solution of $\text{P3}$.

\begin{thm} \label{thm1}
  If $\bar{\mathbf{x}},\bar{\mathbf{d}}$ is the optimal solution to $\text{P4}$, then $\bar{\mathbf{x}},\bar{\mathbf{d}}$ is the KKT point of $\text{P3}$.
\end{thm} 
  
\begin{proof}
  Proof is provided in Appendix B. 
\end{proof}

\begin{thm} \label{thm2}
  The problem $\text{P3}$ can get a stationary solution by Algorithm \ref{alg:SCA}.
\end{thm} 

\begin{proof}
  Proof is provided in Appendix C. 
\end{proof}

According to Theorem \ref{thm1} and Theorem \ref{thm2}, we can get that Algorithm \ref{alg:SCA} can converge to a stationary point and the point is the KKT point of $\text{P3}$. Since the original problem is non-convex, the global optimal solution cannot be obtained. The solution obtained by Algorithm \ref{alg:SCA} based on the SCA method \cite{razaviyaynSuccessiveConvexApproximation2014} is the approximate optimal solution of the original problem. 

\begin{figure}[tbp]
	\begin{algorithm}[H]
		\caption{Successive convex approximation algorithm}
    \label{alg:SCA}
    \begin{algorithmic}[1]
      \STATE{Find a feasible point $\mathbf{x}^0$ and $\mathbf{d}^0$, choose a stepsize $\alpha \in (0,1]$, and set $r = 0$, $\epsilon > 0$}
      \REPEAT
      \STATE{$z_x^{r+1},z_d^{r+1} = \arg \min U_{qp}(\mathbf{x,d};\mathbf{x}^r,\mathbf{d}^r)$}
      \STATE{$x^{r+1} = x^{r} + \alpha(z_x^{r+1} - x^{r})$}
      \STATE{$d^{r+1} = d^{r} + \alpha(z_d^{r+1} - y^{r})$}
      \STATE{$r \leftarrow r + 1$}
      \UNTIL{$\Vert \mathbf{x}^{r} - \mathbf{x}^{r-1} \Vert + \Vert \mathbf{d}^{r} - \mathbf{d}^{r-1} \Vert \leqslant \epsilon$}
    \end{algorithmic}
	\end{algorithm}
\end{figure}

\section{Resource Reallocation Scheme} \label{reallocation}

We can get a microservice deployment strategy for layer sharing and chain sharing from Sec. \ref{Solution1}. However, the computing resources of all servers will not be fully utilized due to the constraints of computing resources and storage resources of edge servers. Edge servers may face the problem that one resource is used up while the other resource is still available. It is a waste of spare resources. In this section, we will propose a server resource redistribution method, which can fully utilize the spare resources of the server.

\subsection{Problem formulation}

After a microservice is deployed, the deployment location of the microservice remains unchanged until the end of the microservice. Assume that the computing resource of server $n$ is $C_n^C$, and $J$ microservices are deployed on it. These microservices can be described by a set $\mathbb{J} = \{1,\cdots,J\}$. The minimum computing resource requested by microservice $j$ is $u_j$, and the computing resource actually allocated to microservice $j$ is $f_j$. The computing resources allocated to the microservice must be higher than the computing resources it requests, so there is a constraint $
f_j \geqslant u_j$.

Assuming that the amount of data that the microservice needs to process is $Data$. The original processing time required is $t_{old} = \frac{Data}{u_j}$, and the new processing time is $t_{new} = \frac{Data}{f_j}$. The ratio of the new processing time to the original processing time is $\frac{t_{new}}{t_{old}} = \frac{u_j}{f_j}$. So we define the evaluation function $e_j = \frac{u_j}{f_j}$ to evaluate the impact of allocated computing resources on the processing efficiency of microservices. Then we can define the optimization problem as follows

\begin{align}
  \text{P5:} \min U =  & \quad \sum_{j \in \mathbb{J}}e_j \\
  \text{s.t.} \quad& f_j \geqslant u_j, \forall j \in \mathbb{J} \\
  & \sum_{j \in \mathbb{J}} f_j \leqslant  C_n^C \label{resource}
\end{align}

Constraint (\ref{resource}) means that the total computing resources allocated need to be less than the total resources of the server.

\begin{figure}[tbp]
	\begin{algorithm}[H]
		\caption{Greedy deployment strategy}
    \label{alg:Simple}
    \begin{algorithmic}[1]
      \STATE{Normalize the communication data and layer size of each microservice by: $w^{k}_{ij,new} = \theta\frac{w^{k}_{ij} - w_{min}}{w_{max} - w_{min}}$, $S^l_{new} = (1-\theta)\frac{S^l - S_{min}}{S_{max} - S_{min}}$}
      \STATE{Sort $w^{ki}_{new}$ and $S^l_{new}$. The larger the value is, the higher the priority will be. Each value corresponds to two microservices with a large amount of communication data or several groups of microservices with a larger image layer. Get the sorted list $List$}
      \FOR{$ms$ in $List$ }
        \STATE{Check if these microservices in $ms$ have been deployed}
        \IF{all microservices have been deployed}
          \STATE{continue}
        \ELSE
          \IF{microservices in $ms$ are from $w^{k}_{ij,new}$}
            \STATE{Find the closest server $n$ where the microservice is deployed and deploy microservice in $n$.}
            \IF{server $n$ has no enough resource}
              \STATE{find a closet server from server $n$ to deploy.}
            \ENDIF

          \ELSE
            \STATE{Deploy microservice in the server which has maximum $b_n^{cloud}$ and enough resource}
          \ENDIF
        \ENDIF
      \ENDFOR
      \STATE{Output deployment strategy}
    \end{algorithmic}
	\end{algorithm}
\end{figure}

\subsection{Solution based on Lagrange Multiplier Method}

The Lagrangian function of $\text{P5}$ is constructed as

\begin{equation}
L(f_j,\lambda_i,\mu) = \sum_{j \in \mathbb{J}}e_j + \sum_{j \in \mathbb{J}}\lambda_i(u_j-f_j)+ \mu(\sum_{j \in \mathbb{J}} f_j - C_n^C)
\end{equation}

Its KKT condition is

\begin{align}
  \begin{cases}
  & \nabla L_{f_j}(f_j,\lambda_i,\mu) \\ &\quad= -\sum_{j \in \mathbb{J}}\frac{u_j}{f_j^2} + \sum_{j \in \mathbb{J}}\lambda_if_j+ \mu J = 0 \\
  & \lambda_i(u_j-f_j) = 0, \forall i \in \mathbb{J} \\
  & \mu(\sum_{j \in \mathbb{J}} f_j - C_n^C) = 0 \\
  & \lambda_i  \geqslant 0, \forall i \in \mathbb{J} \\
  & \mu \geqslant 0
  \end{cases} \label{KKT}
\end{align}

By solving (\ref{KKT}), we can get

\begin{equation}
  f_j = \frac{u_j}{\sum_{j \in \mathbb{J}}u_j}C_n^C
\end{equation}

A simple explanation is that each microservice is proportionally multiplied by the ratio of the total computing resources to the initial request resources. In this way, we can redistribute the computing resources and make full use of the computing resources.

\section{Performance Evaluation} \label{performance}

In this section, we evaluate the performance to verify the effectiveness of our proposed method. We use \emph{Gurobi} \cite{GurobiFastestSolver} to solve the integer quadratic programming problem $\text{P4}$. The proposed method is compared with five methods: 

\begin{itemize}
  \item Greedy Deployment Strategy\cite{guLayerawareCollaborativeMicroservice2022} (GDS): A deployment strategy based on the greedy strategy in \cite{guLayerawareCollaborativeMicroservice2022}. We modified it to fit the experiments in this paper. This algorithm works by weighting the size of the layer and chain. The steps of the algorithm are shown in Algorithm \ref{alg:Simple}.
  \item Layer-match Scheduling\cite{fuFastEfficientContainer2020} (LS): For each microservice, select an edge server with most amount of its image layers stored locally and sequence layers according to the assignment order.
  \item Kubernetes Deployment Strategy\cite{KubernetesKubernetes} (K8S): Kubernetes default deployment policy schedules microservices to edge servers with the required images stored locally, otherwise, to the edge server with the least total download size.
  \item Layer-sharing Deployment Strategy (LDS): A deployment strategy that only considers layer sharing. It is a special case of the proposed method when $\theta=1$.
  \item Chain-sharing Deployment Strategy (CDS): A deployment strategy that only considers chain sharing. It is a special case of the proposed method when $\theta=0$.
\end{itemize}

We use Python to conduct simulation experiments on multiple servers to evaluate the performance of the proposed method under different conditions. To accurately evaluate the effectiveness in the real world, we conduct experiments with five real edge servers. Furthermore, we carried out large-scale simulation tests on 15 servers to evaluate the adaptability of the algorithm in large-scale scenarios.


\subsection{Simulation experiment}

\subsubsection{Experimental environment}

\begin{figure*}[t]
  \centering
  \subfigure[Different production scale]{
    \centering
    \includegraphics[width=0.31\textwidth]{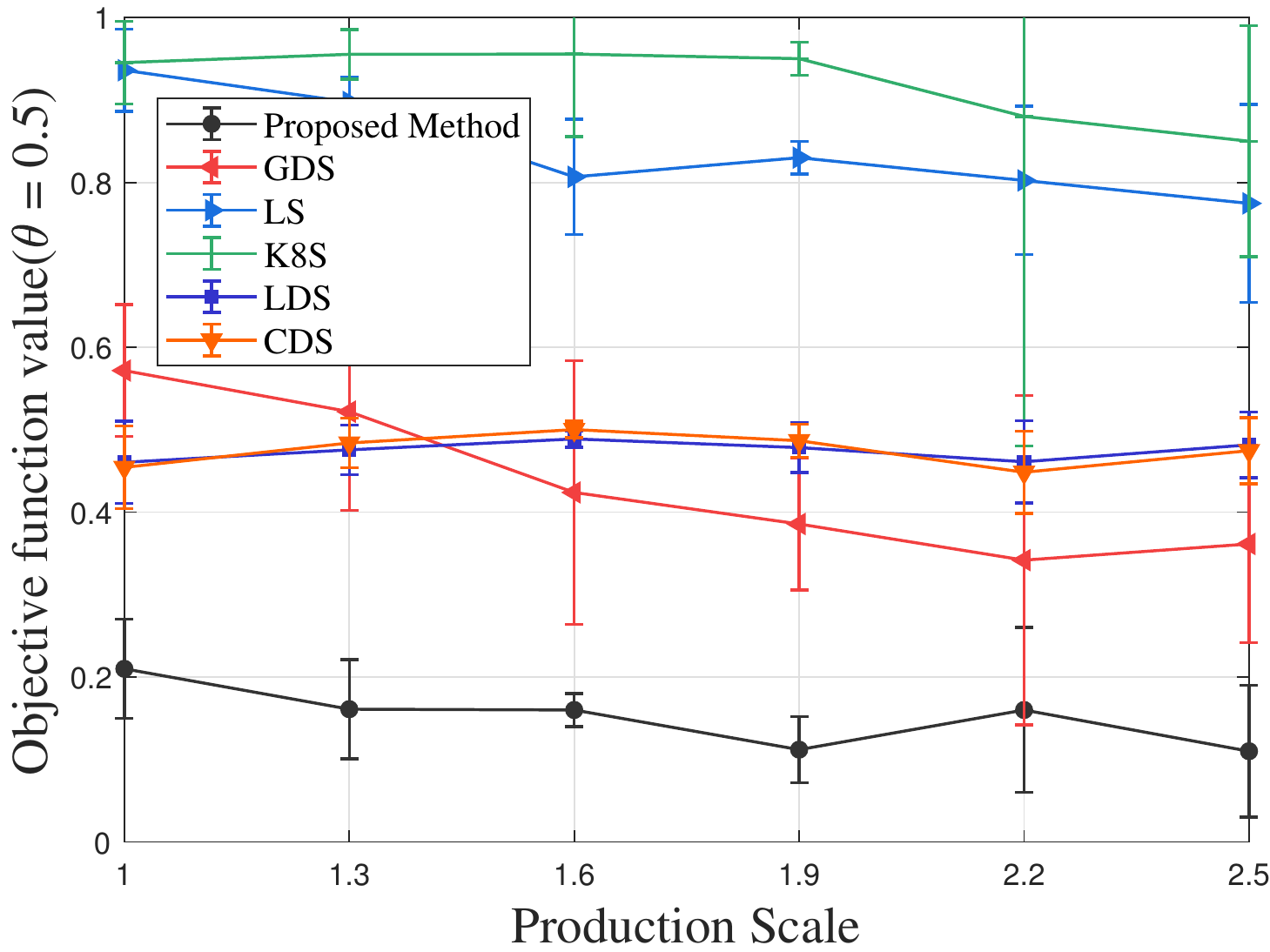}
    \label{experiment1_1_1}
  }
  \subfigure[Different microservices]{
    \centering
    \includegraphics[width=0.31\textwidth]{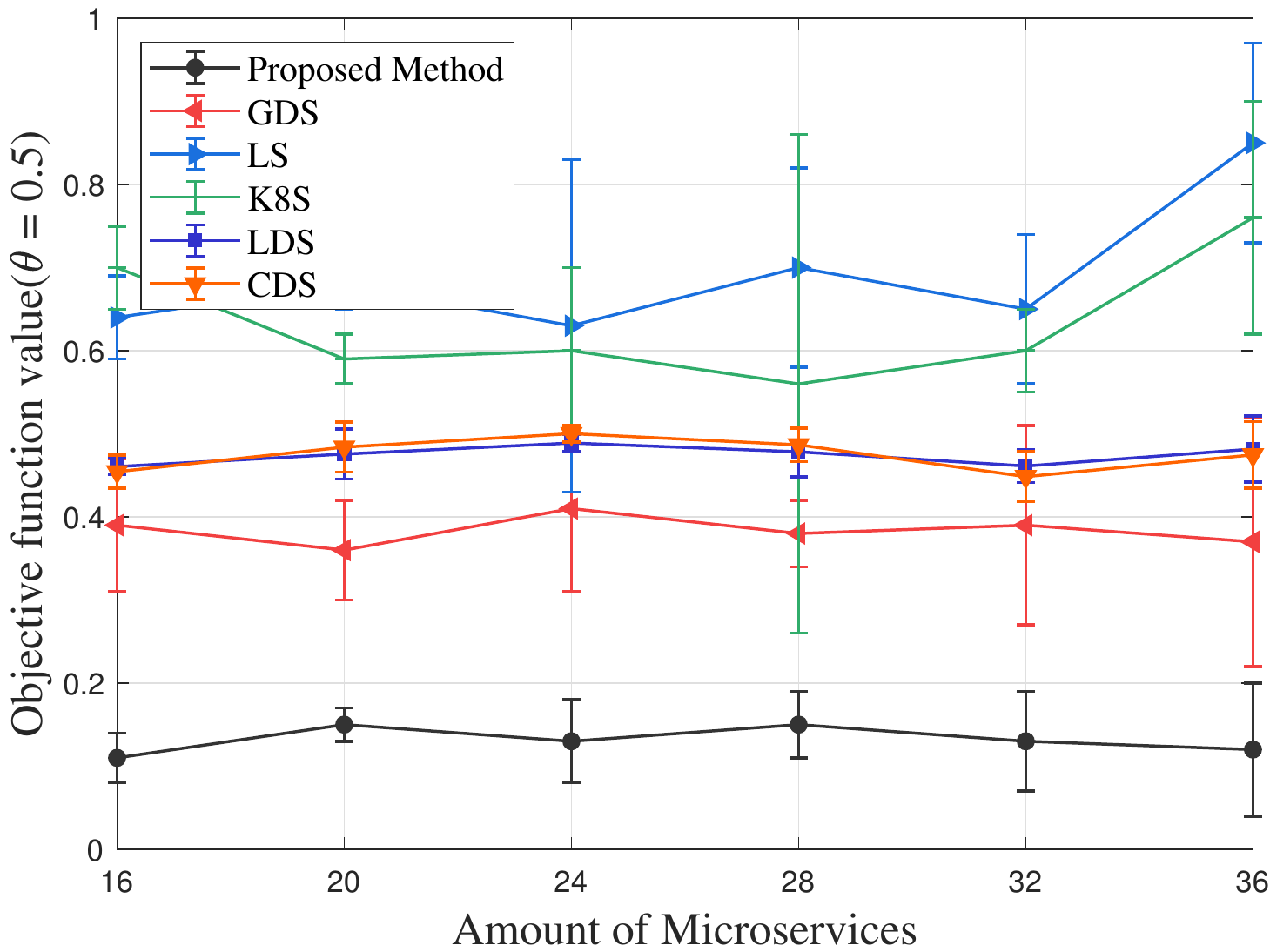}
    \label{experiment1_1_2}
  }
  \subfigure[Different $\theta$]{
    \centering
    \includegraphics[width=0.31\textwidth]{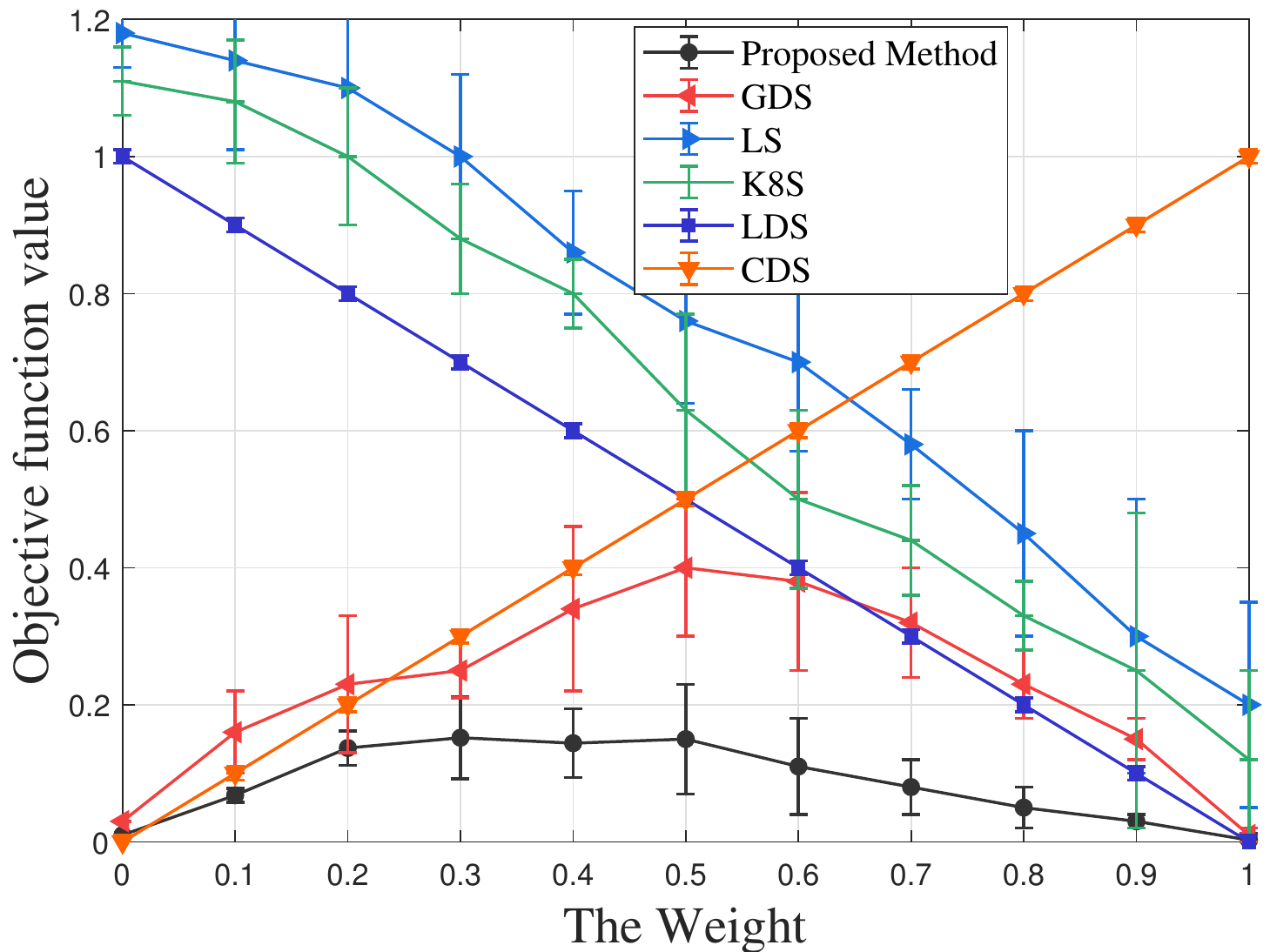}
    \label{experiment1_1_3}
  }
  \caption{Objective function value with different condition}
  \label{experiment1_1}
\end{figure*}

\begin{figure*}[t]
  \centering
  \subfigure[Total image pull delay with different $\theta$]{
    \centering
    \includegraphics[width=0.4\textwidth]{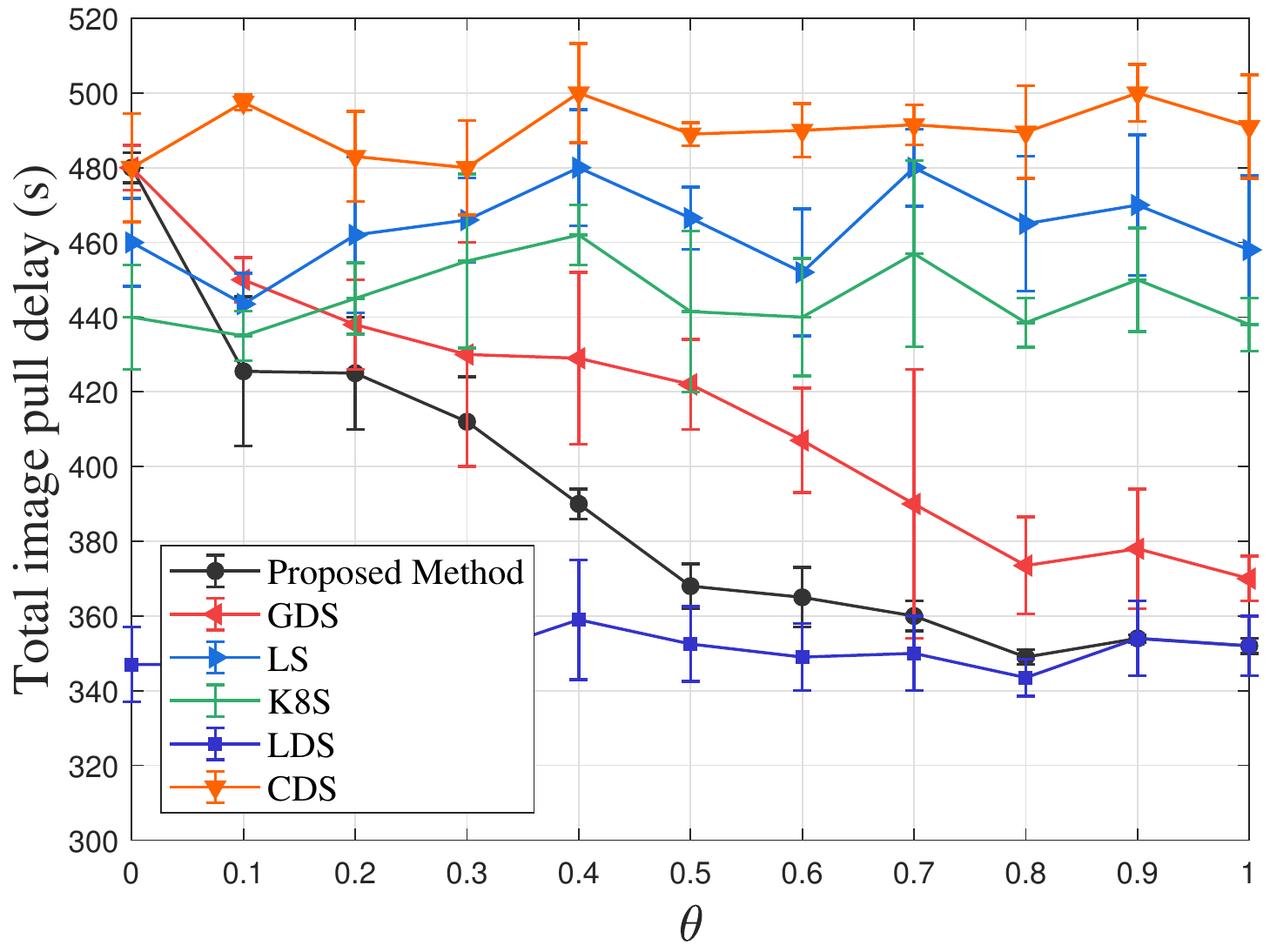}
    \label{experiment1_2_1}
  }
  \subfigure[Total communication overhead with different $\theta$]{
    \centering
    \includegraphics[width=0.4\textwidth]{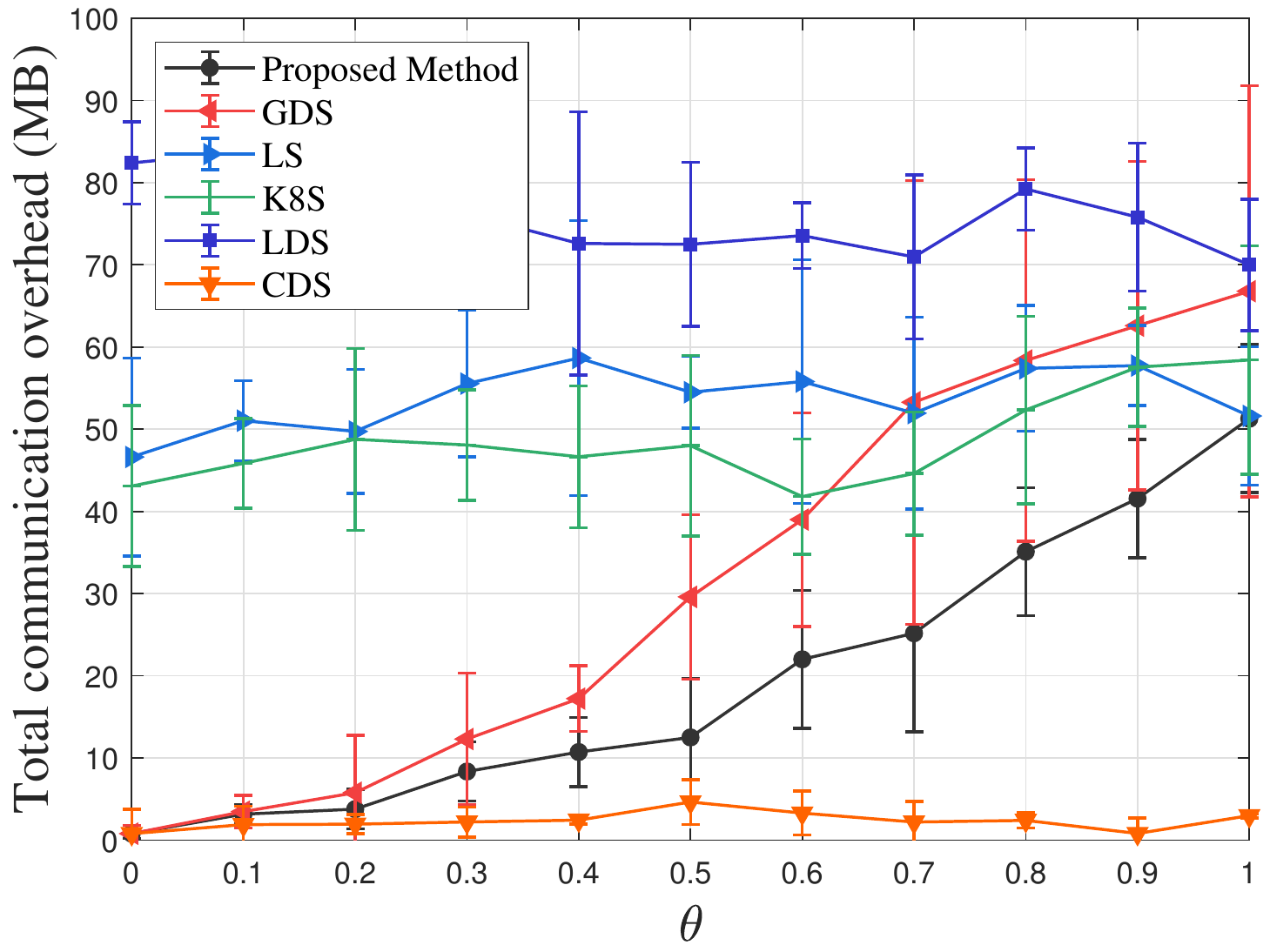}
    \label{experiment1_2_2}
  }
  \caption{Total image pull delay and total communication overhead with different $\theta$}
  \label{experiment1_2}
\end{figure*}

The experimental platform is Python 3.9.12. The experiments are carried out on a CentOS 7 system equipped with Inter 4210R, 2.40GHz, and 64 GB RAM. We simulated a smart manufacturing production scenario with up to 9 edge servers and 36 microservices. The average storage resource of edge servers is 8 GB, the average computing resource (CPU frequency) of each server is 1.8 GHz with 4 cores, and the bandwidth between the server and the cloud server is 80-200 MB/s. The hop of adjacent servers is 1, and the element values of the $\mathbf{D}$ matrix vary from 0 to 5, which means the maximum hop of servers is 5. Each application consists of 2-6 microservices, and the communication data between microservices ranges from 100-2000 KB. The computing resource requirements of each microservice range from 0.002 GHz to 1.0 GHz. The number of layers of each microservice is in the range of 6-13. By dividing these layers into shareable and unshareable layers, each image can be regarded as a microservice composed of 1-2 layers. This can reduce the difficulty of calculation. The size of each layer varies from 1-1220 MB. The above data are randomly generated in each experiment to verify the stability of the proposed method. The specific experimental parameter settings are shown in Table \ref{tparameter}.

\begin{table}[htbp]
	\centering 
	\caption{Experimental parameter}
	\label{tparameter}
	\begin{tabular}{p{35pt}|p{60pt}} \hline \hline
    \rule{0pt}{8pt}\textbf{Symbol} & \textbf{Value}\\ \hline
    \rule{0pt}{8pt} $K$ & 4-9 \\ \hline
    \rule{0pt}{8pt} $N$ & 4-9 \\ \hline
    \rule{0pt}{8pt}$A_k$ &2-6 \\ \hline
    \rule{0pt}{8pt}$C_n^C$ &1.4-2.2 GHz \\ \hline
    \rule{0pt}{8pt}$C_n^S$ &4-16 GB \\ \hline
    \rule{0pt}{8pt}$b_n^{cloud}$ &120-200 MB/s \\ \hline
    \rule{0pt}{8pt} $S^l$ &1-1220 MB \\ \hline
    \rule{0pt}{8pt}$u^{ki}$ &0.002-1.0 GHz \\ \hline
    \hline
	\end{tabular}
\end{table}

\subsubsection{Experimental results}

Fig. \ref{experiment1_1} shows the objective function value of different methods. Fig. \ref{experiment1_1_1} shows the objective function value in different production scale. We take the minimum amount of microservices $ms_{min} = 12$ and the amount of servers $n_{min} = 4$ in this experiment as the benchmark values. Then the scale can be described as $Scale = \frac{1}{2}(\frac{N_{ms}}{ms_min} + \frac{N_{n}}{n_min})$, where $N_{ms}$ is the amount of microservices, and $N_n$ is the amount of servers. As can be seen from the figure, our proposed deployment strategy can minimize the objective function compared with the other five methods. The proposed method can also reach relatively stable results under different numbers of microservices and servers. The GDS method can also achieve a good deployment strategy by weighting the layer size and communication overhead. However, due to its greedy strategy, the optimal solution may not always be obtained. The LDS method and the CDS method cannot make the objective function optimal because they only consider one aspect of resource sharing. The LS and K8S methods only consider layer sharing and can not get a better result due to the high communication overhead.


Fig. \ref{experiment1_1_2} shows the objective value with different microservices and nine servers. It simulates different production loads. The higher the number of microservices is, the higher the load on one server will be. We can see that the proposed deployment strategy can achieve the optimal objective function value, and the value fluctuates within a small range under different microservice loads, which shows that the proposed method is suitable for different load conditions and has good stability. The results of other methods are worse than the proposed method.

\begin{figure}[tbp]
  \subfigure[Figure of edge servers]{
    \centering
    \includegraphics[width=0.23\textwidth]{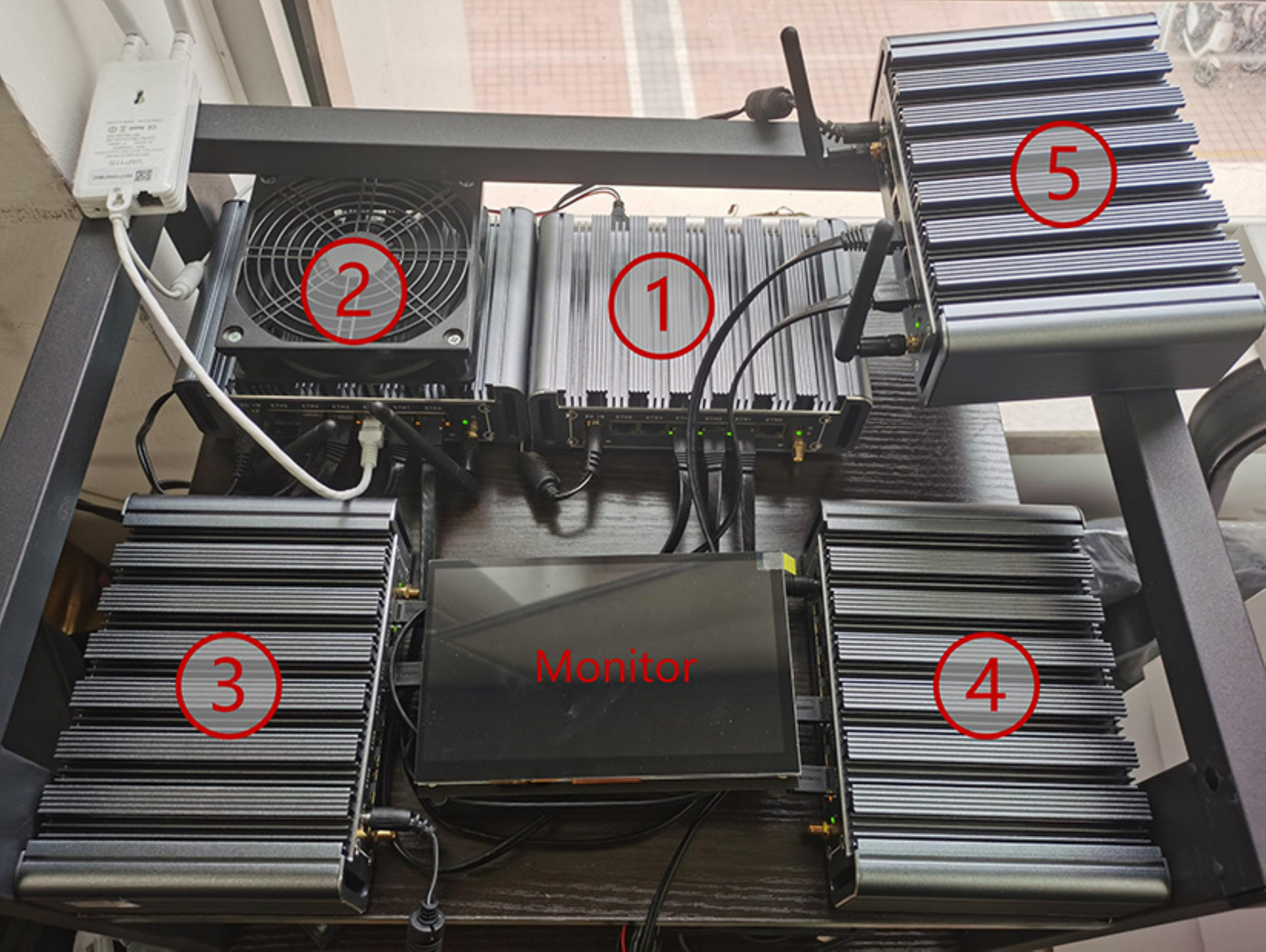}
    \label{experiment2_1_1}
  }
  \subfigure[Topology of edge servers]{
    \centering
    \includegraphics[width=0.23\textwidth]{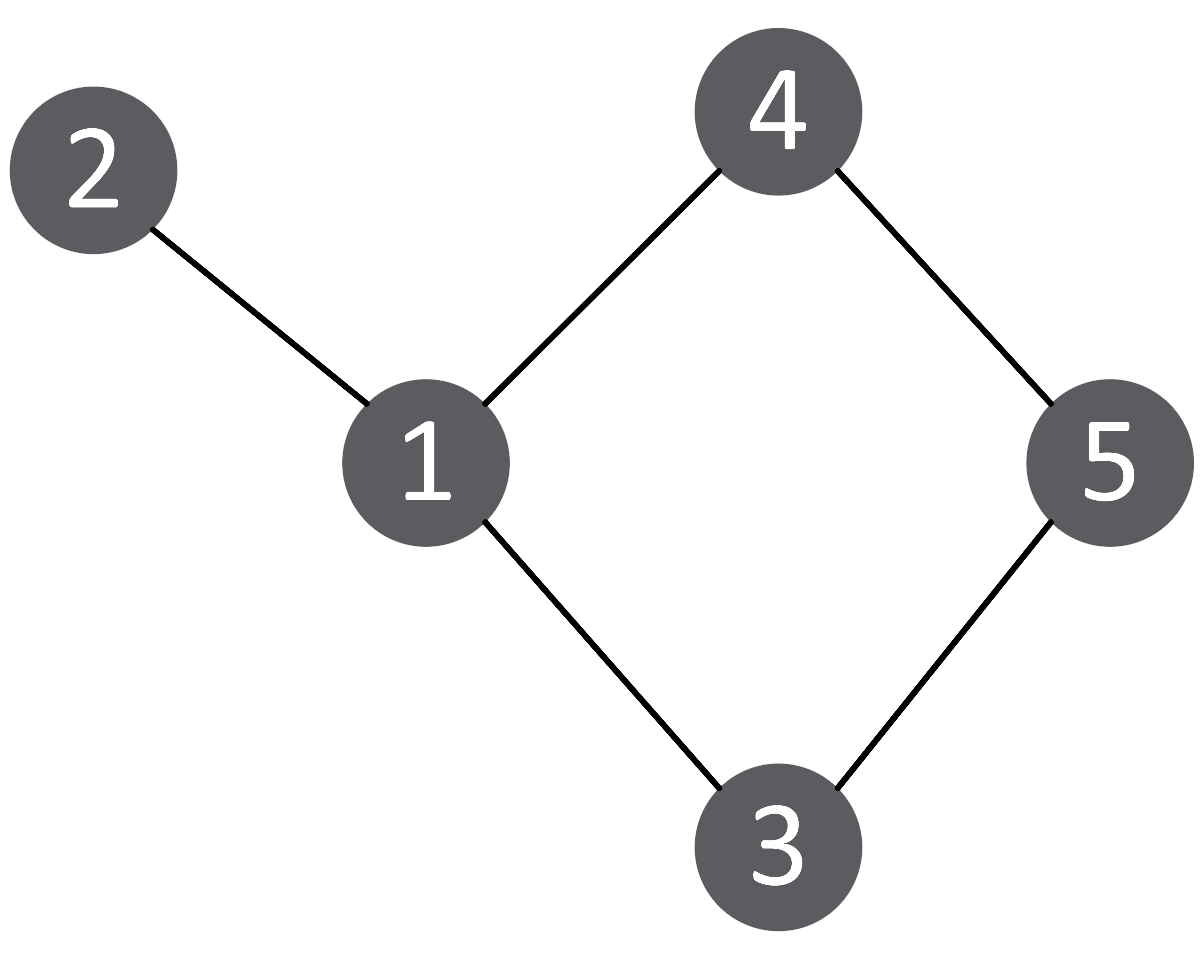}
    \label{experiment2_1_2}
  }
  \caption{Figure and topology of edge servers}
  \label{experiment2_1}
\end{figure}

\begin{figure}[t]
  \subfigure[Delay with different storage capacity]{
    \centering
    \includegraphics[width=0.23\textwidth]{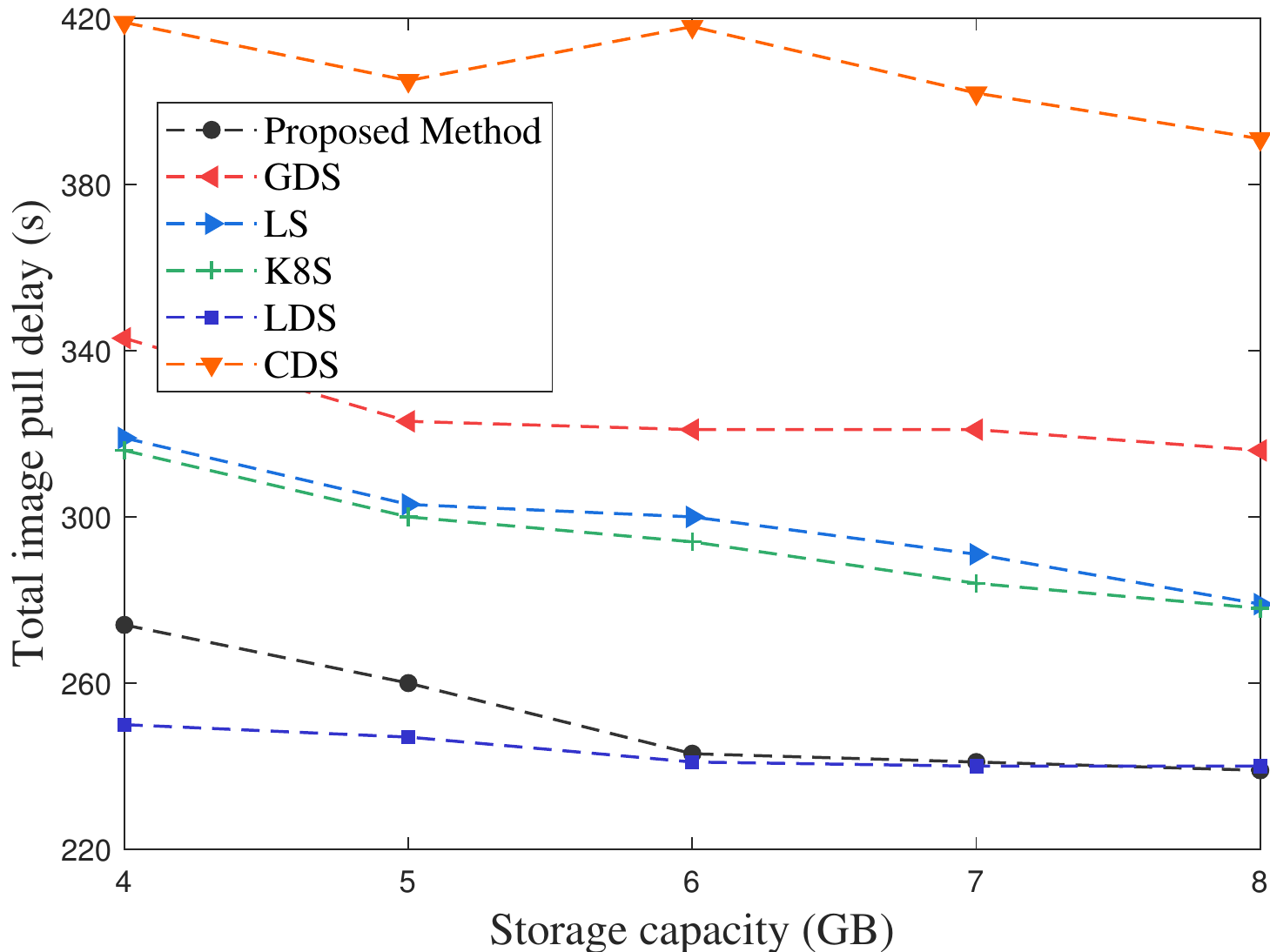}
    \label{experiment2_2_1}
  }
  \subfigure[Overhead with different storage capacity]{
    \centering
    \includegraphics[width=0.23\textwidth]{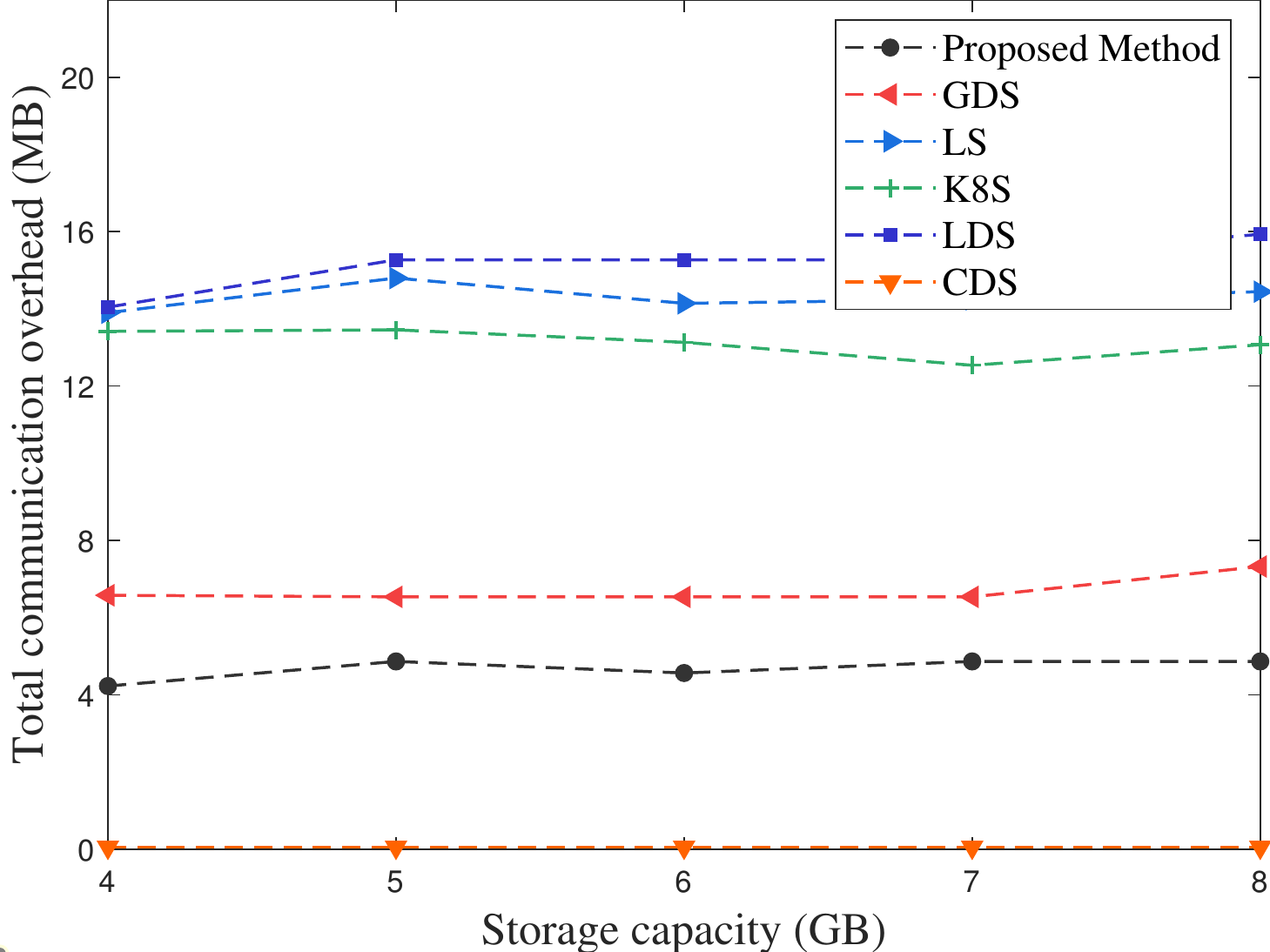}
    \label{experiment2_2_2}
  }
  \caption{Delay and overhead with different storage capacity}
  \label{experiment2_2}
\end{figure}

\begin{figure}[t]
  \subfigure[Delay with different computing capacity]{
    \centering
    \includegraphics[width=0.23\textwidth]{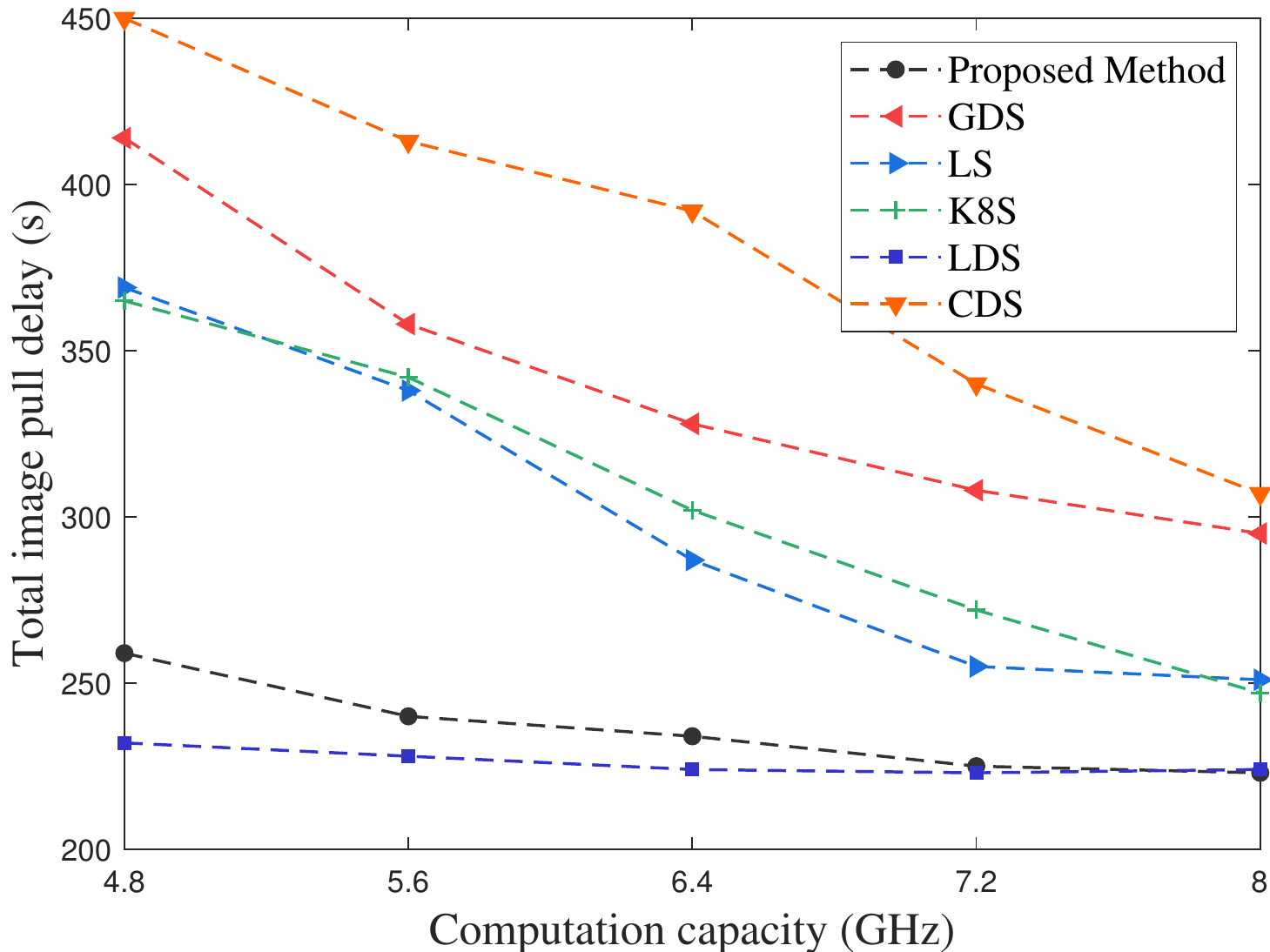}
    \label{experiment2_2_3}
  }
  \subfigure[Overhead with different computing capacity]{
    \centering
    \includegraphics[width=0.23\textwidth]{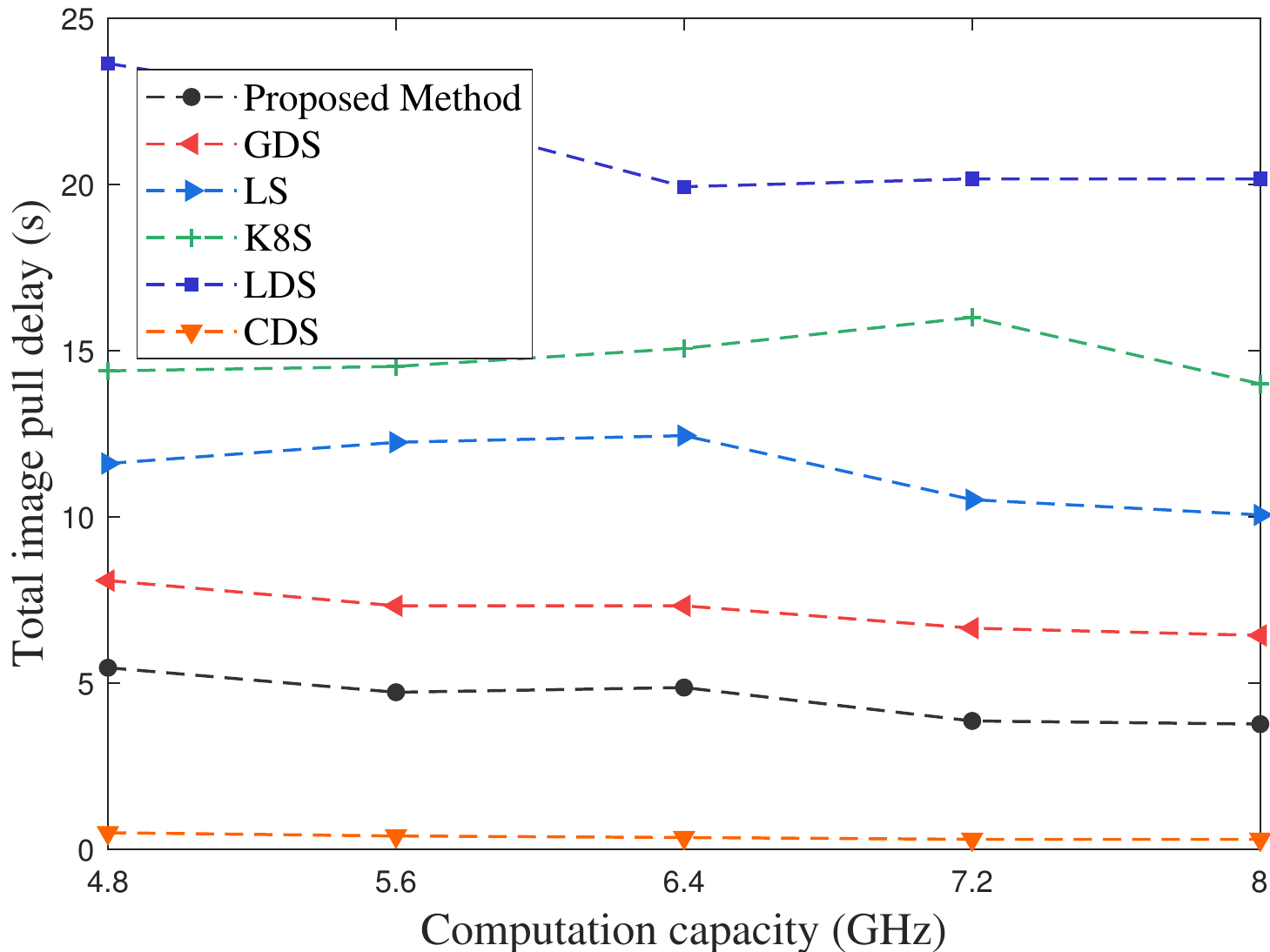}
    \label{experiment2_2_4}
  }
  \caption{Delay and overhead with different computing capacity}
  \label{experiment2_3}
\end{figure}

\begin{figure}[t]
  \centering
  \subfigure[$\#1$]{
    \centering
    \includegraphics[width=0.23\textwidth]{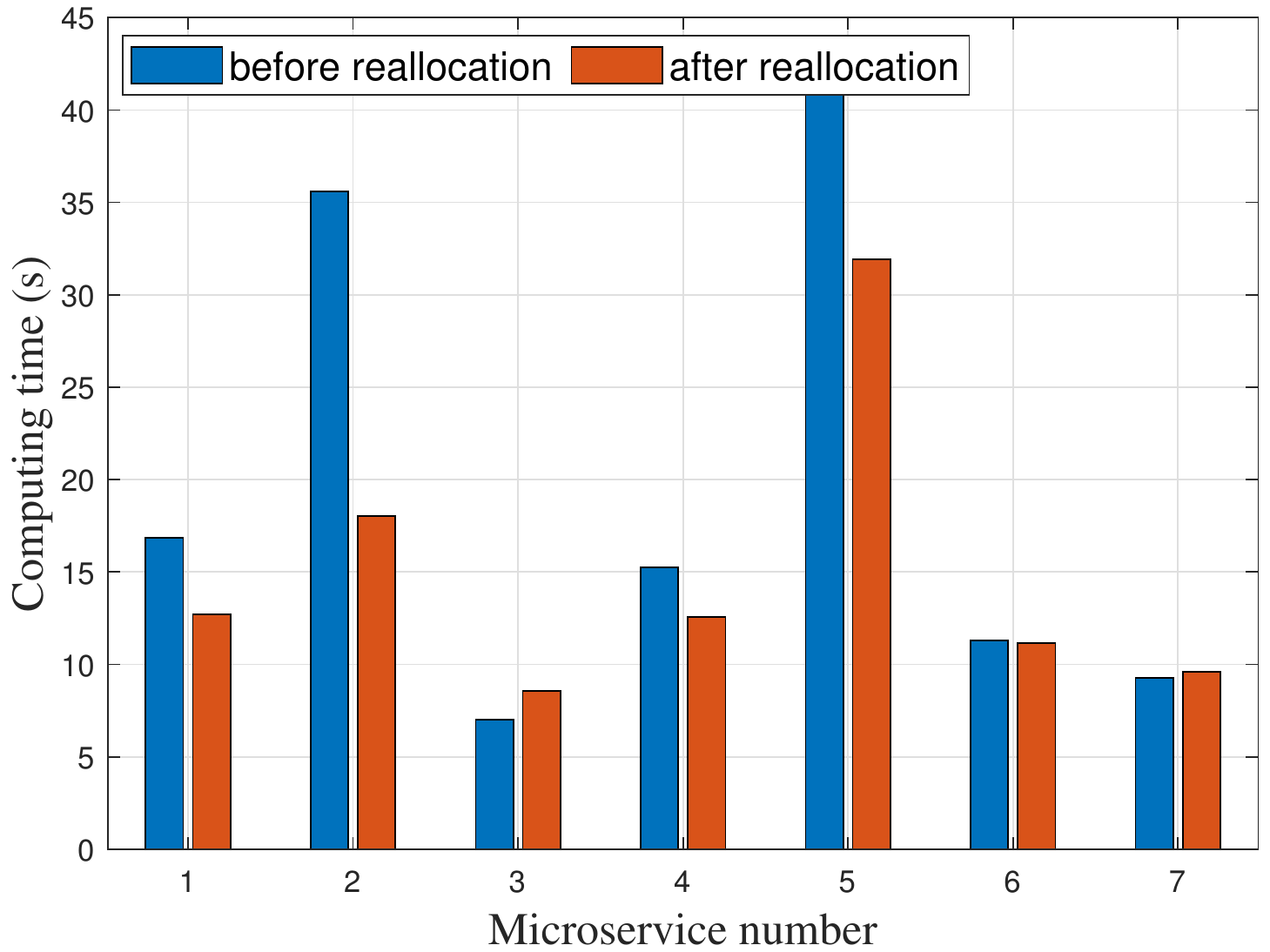}
    \label{experiment2_4_1}
  }
  \subfigure[$\#2$]{
    \centering
    \includegraphics[width=0.23\textwidth]{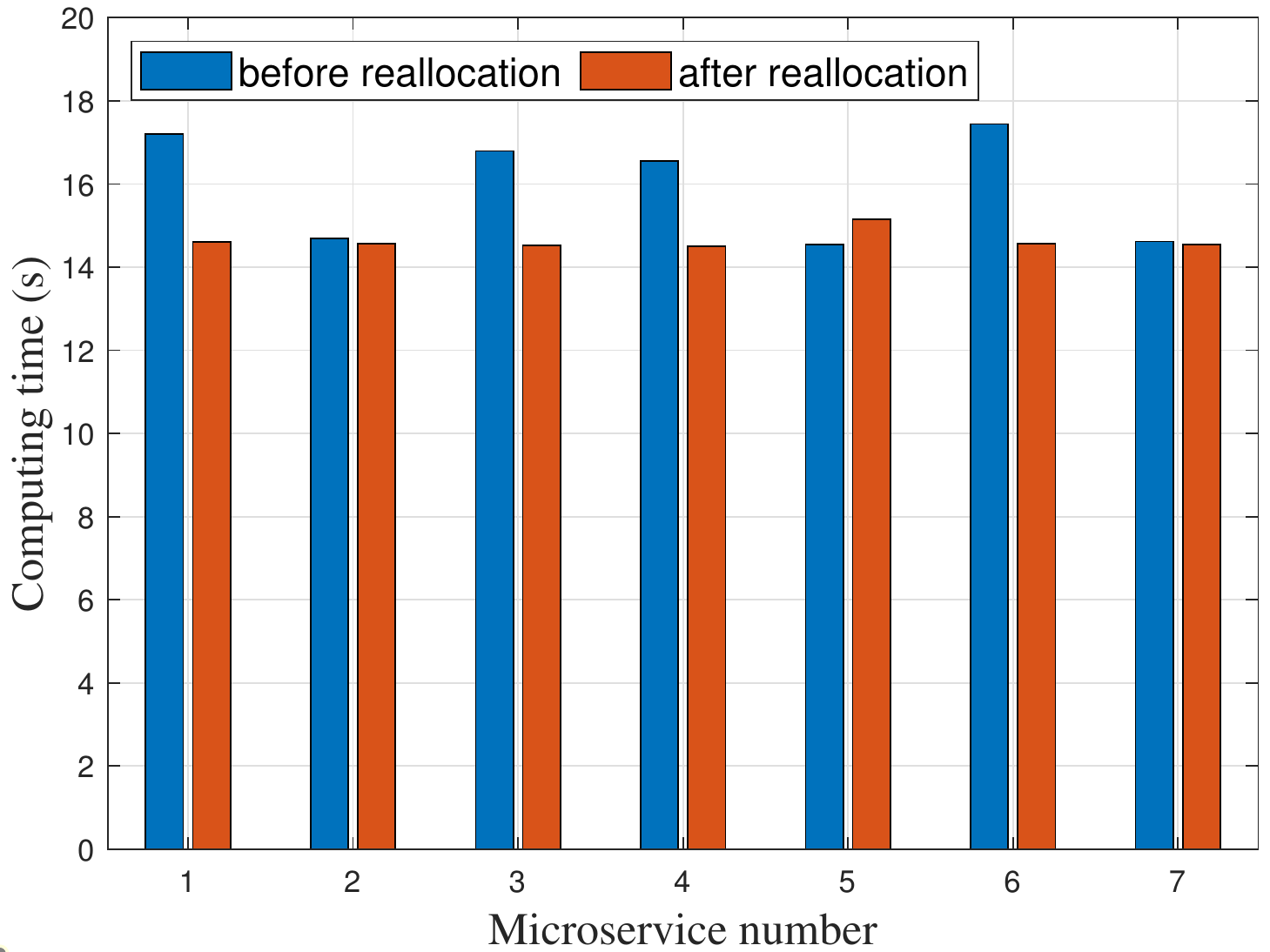}
    \label{experiment2_4_2}
  }
  \caption{The impact of resource reallocation strategy on the computing time of microservices}
  \label{experiment2_4}
\end{figure}

Fig. \ref{experiment1_1_3} shows the objective value with different $\theta$ when there are 9 servers and 36 microservices on average. It simulates the effect of different weights for microservice image pull delay and communication overhead. In this figure, the function values of the LDS method, LS method, and K8S method change linearly because they only consider layer sharing of the objective function. And CDS method only considers chain sharing of the objective function. So the changes of $\theta$ can not impact the deployment strategy. The proposed method can achieve optimal results no matter what value $\theta$ takes. When $\theta = 0$, there is only the chain sharing part in the objective function, and the objective function value is the same as the CDS method. When $\theta = 1$, there is only the layer sharing part in the objective function, and the objective function value is the same as the LDS method.

Fig. \ref{experiment1_2} shows the image pull delay and communication overhead with different $\theta$. Since the LDS, CDS, LS, and K8S methods are unaffected by $\theta$, the values of these two strategies do not change a lot in the two figures. The fluctuation of the line is more due to randomly generated microservice data. Fig. \ref{experiment1_2_1} shows the total image pull delay with different $\theta$. The higher the weight $\theta$, the lower the image pull delay of both the proposed strategy and the GDS method. When $\theta = 1$, the proposed deployment strategy can achieve the same result as the LDS method. The image pull delay can be reduced by 140s compared to the CDS method. The proposed strategy can reduce the total image pull delay by 52s on average compared to the GDS method. Fig. \ref{experiment1_2_2} shows the total communication overhead with different $\theta$. The lower the weight $\theta$, the lower the total communication overhead of both the proposed strategy and the GDS method. When $\theta = 0$, the proposed deployment strategy can achieve the same result as the CDS method. The total communication overhead can be reduced by 80 MB compared to the LDS method. The proposed strategy can reduce total communication overhead by 10 MB on average compared to the GDS method.

\subsection{Experiment with Real Edge Servers}

\subsubsection{Experimental environment}

We further conduct experiments with five edge servers in real world to evaluate the effectiveness of our method. The servers are shown in Fig. \ref{experiment2_1_1} and the topology of servers is shown in Fig. \ref{experiment2_1_2}. Each server has an i5-8250U CPU, 8G RAM, and is equipped with Docker CE. Communication between servers is carried out using the TCP protocol. We select 23 microservices from Docker Hub. The image size of these microservices is in the range of 1.24-1098 MB, and the computation resource requirement is in the range of 0.2-1.4 GHz. The number of layers of each microservice is in the range of 4-11. We simulate servers with different storage and resource constraints by limiting the resource usage of docker. The experimental results are shown in the following figures. Each data point in the figures is the average of multiple experiments.

\subsubsection{Experimental results}

Fig. \ref{experiment2_2} shows the total image pull delay and communication overhead under different storage capacities with $\theta = 0.5$. As can be seen from Fig. \ref{experiment2_2_1}, the results of the proposed method and the LDS method are very close. When storage capacity becomes more and more sufficient, the proposed method can significantly reduce the image pull delay compared with other methods. This is because when the storage resources are sufficient, more microservices with the same layer can be deployed on one edge server. It can reduce the size of the image layer to be pulled and reduce the delay. In Fig. \ref{experiment2_2_2}, the CDS method can achieve the lowest communication overhead because it is optimized for the microservice chain. Compared with other methods except the CDS method, the proposed method can achieve lower communication overhead.

Fig. \ref{experiment2_3} shows the image pull delay and communication overhead under different computing capacities with $\theta = 0.5$. The trend is the same as that in Fig. \ref{experiment2_2}. From Fig. \ref{experiment2_2} and Fig. \ref{experiment2_3},  we can also find that the LDS method has the best effect on image pull delay and the worst effect on communication overhead, and the CDS method is the opposite. This is the disadvantage of not considering both aspects comprehensively. Different from other methods, our proposed method can always obtain a better solution that can better balance the image pull delay and communication overhead.

\begin{table*}[tbp]
	\centering 
	\caption{Data of Resource Reallocation}
	\label{tabel_reallocation}
  \renewcommand{\arraystretch}{1.5}
	\begin{tabular}{c|c|c|c|c|c|c|c|c|c} \hline \hline
    \makecell[c]{\textbf{Microservice} \\ \textbf{number ($\#1$)}} & \makecell[c]{\textbf{computing} \\ \textbf{resources}\\ \textbf{(GHz)}} &  \makecell[c]{\textbf{time} \\ \textbf{(s)}} & \makecell[c]{\textbf{resources after} \\ \textbf{reallocation} \\ \textbf{(GHz)}} & \makecell[c]{\textbf{time} \\ \textbf{(s)}} & \makecell[c]{\textbf{Microservice} \\ \textbf{number ($\#2$)}} & \makecell[c]{\textbf{computing} \\ \textbf{resources}\\ \textbf{(GHz)}} & \makecell[c]{\textbf{time} \\ \textbf{(s)}} & \makecell[c]{\textbf{resources after} \\ \textbf{reallocation} \\ \textbf{(GHz)}} & \makecell[c]{\textbf{time} \\ \textbf{(s)}}\\ \hline
     1 & 0.4 & 16.84 & 0.576 & 12.715 & 1 & 0.4 & 17.203 & 0.554 & 14.603\\ \hline
     2 & 0.6 & 35.6 & 0.864 & 18.02 & 2 & 1 & 14.687 & 1.385 & 14.563\\ \hline
     3 & 1.2 & 7 & 1.728 & 8.572 & 3 & 0.6 & 16.791 & 0.83 & 14.516\\ \hline
     4 & 0.8 & 15.26 & 1.152 & 12.567 & 4 & 0.8 & 16.558 & 1.1 & 14.508\\ \hline
     5 & 0.4 & 41.4 & 0.576 & 31.934 & 5 & 1.2 & 14.537 & 1.66 & 15.156\\ \hline
     6 & 0.7 & 11.31 & 1.008 & 11.171 & 6 & 0.6 & 17.443 & 0.83 & 14.566\\ \hline
     7 & 0.9 & 9.26 & 1.296 & 9.613 & 7 & 0.6 & 14.617 & 0.83 & 14.544\\ \hline
     total & 5 & 136.67 & 7.2 & 104.592 & total & 5.2 & 111.836 & 7.2 & 102.456\\ \hline
    \hline
	\end{tabular}
\end{table*}

\begin{figure*}[htbp]
  \subfigure[Network topology]{
    \centering
    \includegraphics[width=0.2\textwidth]{figures/experiment4_1_1.pdf}
    \label{experiment4_1_1}
  }
  \subfigure[Objective function values]{ 
    \centering
    \includegraphics[width=0.24\textwidth]{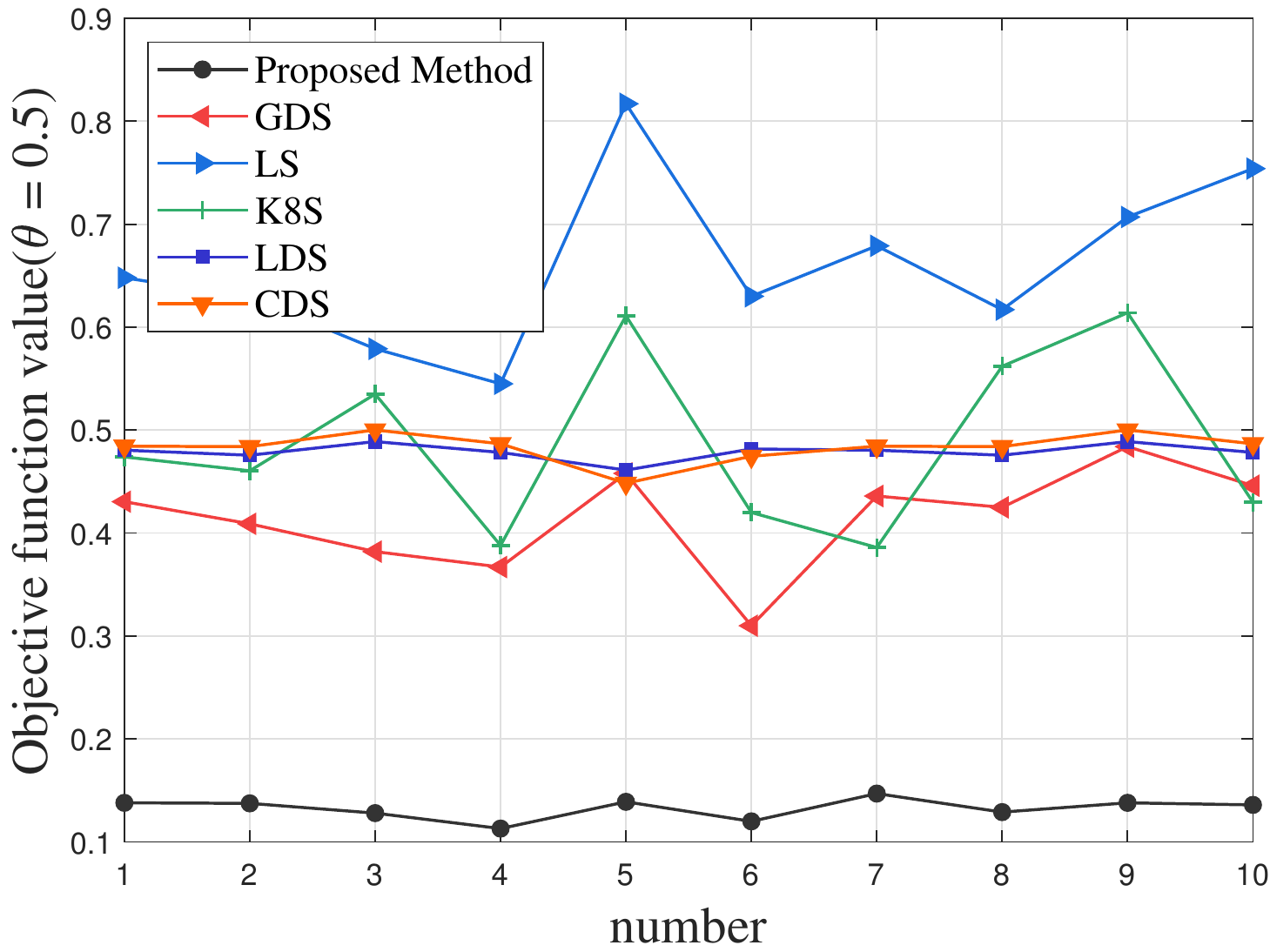}
    \label{experiment4_1_2}
  }
  \subfigure[Total image pull delay]{
    \centering
    \includegraphics[width=0.24\textwidth]{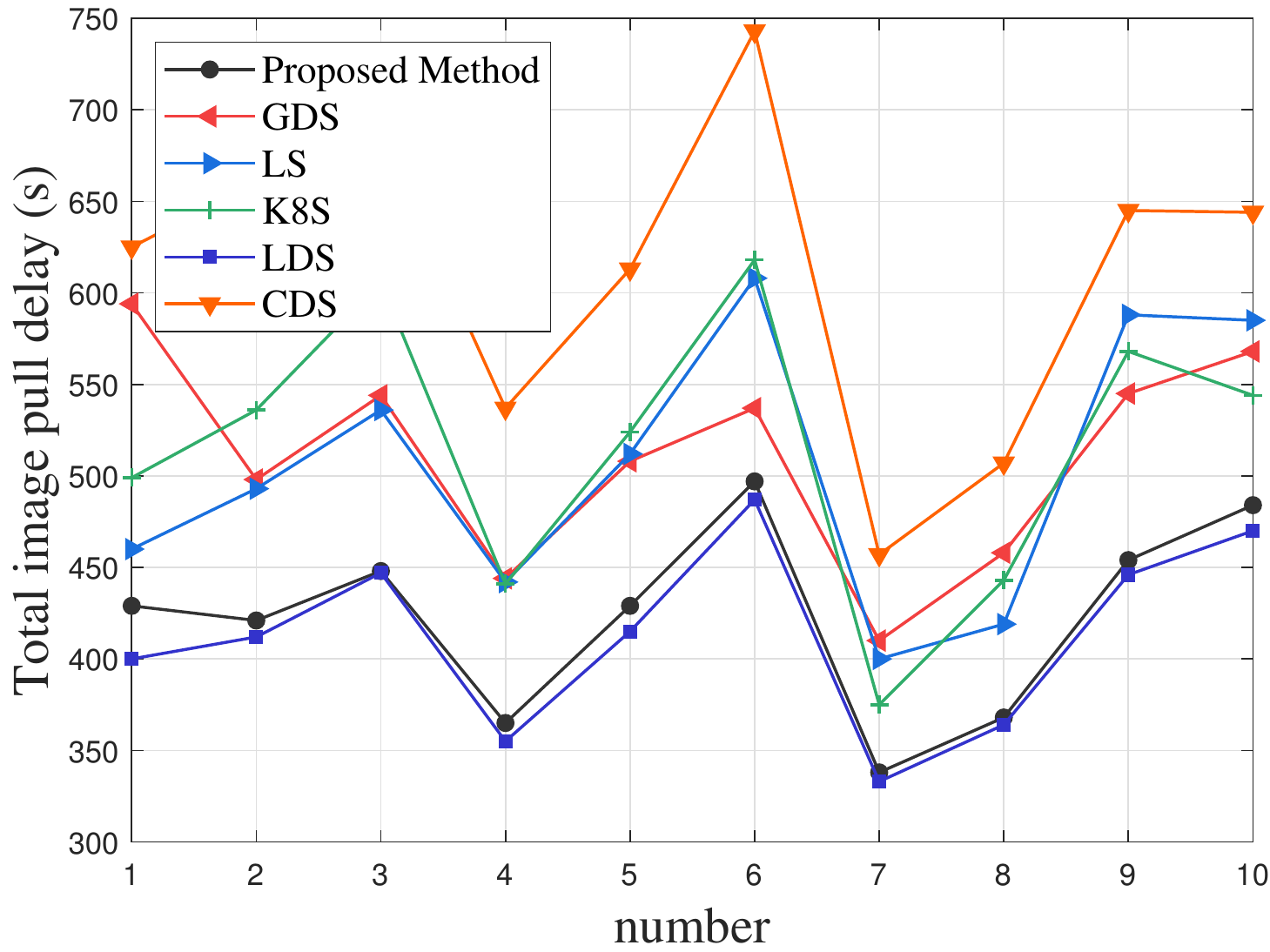}
    \label{experiment4_1_3}
  }
  \subfigure[Total communication overhead]{
    \centering
    \includegraphics[width=0.24\textwidth]{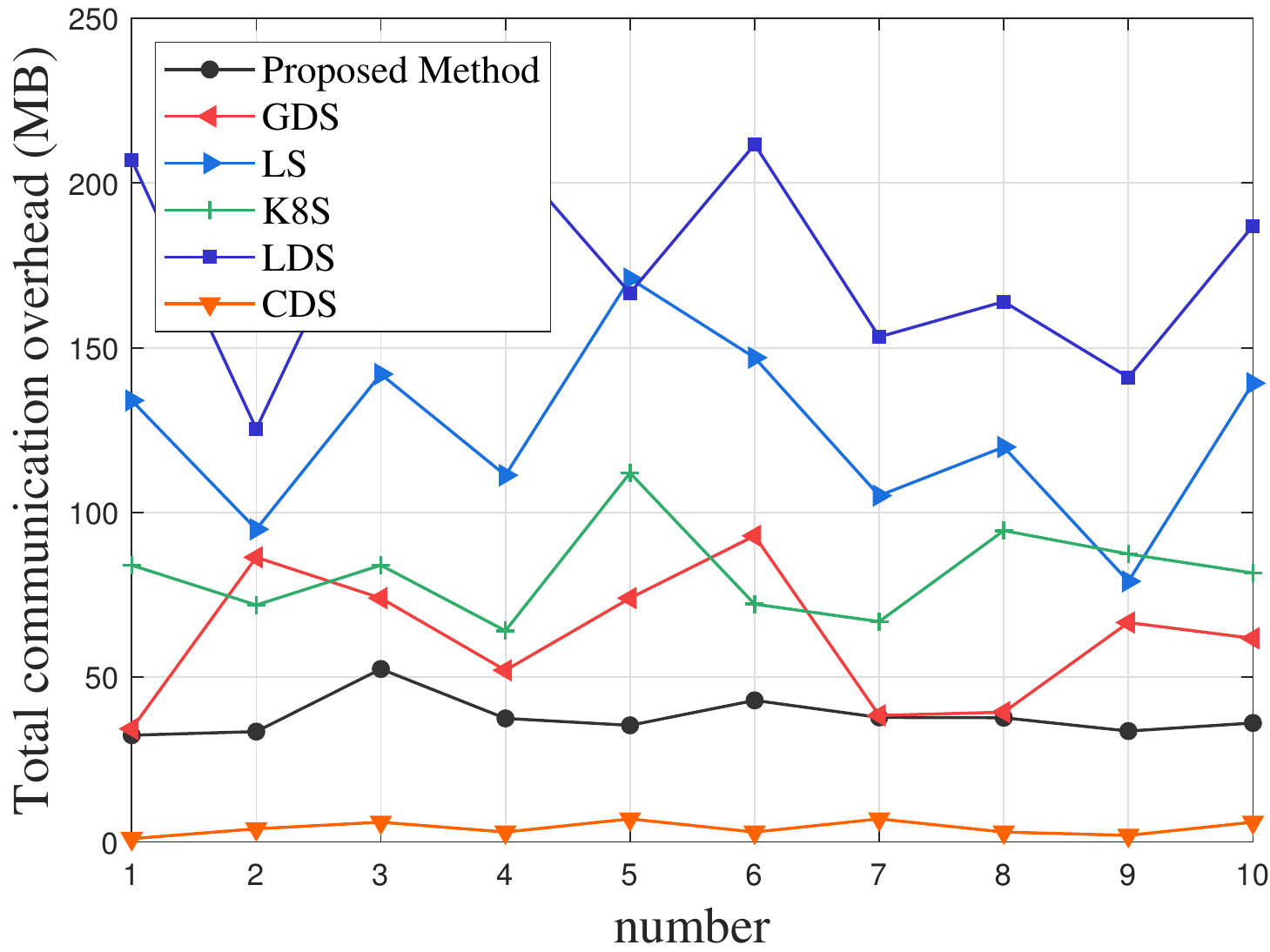}
    \label{experiment4_1_4}
  }
  \caption{Figure and topology of edge servers}
  \label{experiment4_1}
\end{figure*}

Fig. \ref{experiment2_4} shows the impact of the resource reallocation strategy on the completion time of tasks. We selected two from five servers and showed the effect of resource reallocation on their task computing time. The data are shown in Table \ref{tabel_reallocation}. We can see that resource reallocation can significantly reduce the completion time of computing tasks. In Table \ref{tabel_reallocation}, the total task completion time of the two servers is reduced from 136.67 seconds and 111.836 seconds to 104.592 seconds and 102.456 seconds, which is a reduction of 23.7\% and 8.4\%, respectively. This shows that the resource reallocation strategy can effectively reduce the completion time of computing tasks and improve computing efficiency.

\subsection{Large-scale Cases}
To evaluate the performance of our proposed method in a larger scale scenario \cite{guLayerawareCollaborativeMicroservice2022}, we consider the topology of US NSFNET consisting of 15 edge servers as shown in Fig. \ref{experiment4_1_1}. The storage resources of each edge server are 16 GB, the computing resources are 4-core 1.6 GHz, and the bandwidth is 80-120 MB/s. We considered up to 105 microservices selected from Docker Hub and randomly combined them into applications. Other parameter settings are the same as in the previous experiments.

\subsubsection{Experimental results}


To verify the stability of the proposed method, we conducted ten experiments, and the data for each experiment are presented in Fig. \ref{experiment4_1}. Fig. \ref{experiment4_1_2} shows the objective function value for ten experiments, it can be seen that the proposed method has little fluctuation and can achieve the lowest objective function value. Since the microservice composition of each experiment is random, the results vary drastically in Fig. \ref{experiment4_1_3} and Fig. \ref{experiment4_1_4}. The proposed method can achieve almost the same results as the LDS method in terms of image pull delay. It can also achieve similar results to the CDS method in terms of total communication overhead. The results of the proposed method are also better than other schemes.

Fig. \ref{experiment4_2} shows the ratio of image pull delay and communication overhead to the baseline under different servers and different $\theta$. The baseline of image pull delay is defined as the total data size in the absence of layer sharing divided by the average bandwidth of servers. The lower the ratio is, the higher the layer sharing rate will be. As can be seen from the figure, our proposed method can significantly reduce the image pull delay. The image pull delay in the best case is only 56\% of the baseline. When $\theta = 0.5$, an average of 65\% of the baseline ratio can be achieved in the proposed method. The baseline of communication overhead is defined as the summation of all microservice communication data. This ratio can reflect the average number of hops between microservices. The lower the ratio is, the higher the chain sharing rate will be. As can be seen from the figure, our proposed method can significantly reduce communication overhead. The communication overhead in the best case is only 5\% of the baseline. When $\theta = 0.5$, an average of 30\% of the baseline value ratio can be achieved in the proposed method. Moreover, The image pull delay and communication overhead can be reduced significantly in any server numbers, which proves the stability of our proposed method.

\begin{figure}[t]
  \centering
  \includegraphics[width=0.4\textwidth]{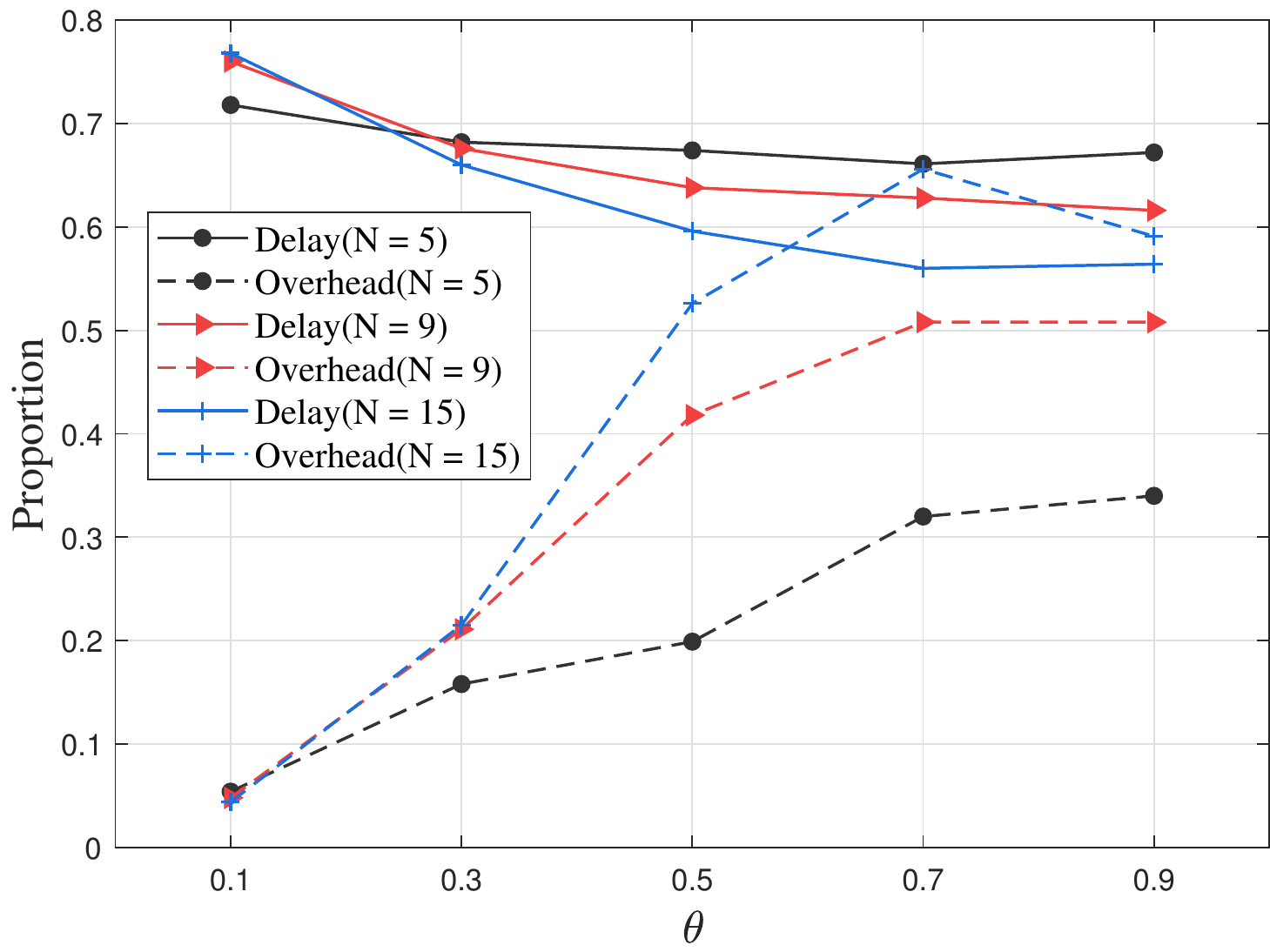}
  \caption{The ratio of image pull delay and communication overhead to baseline in different condition}
  \label{experiment4_2}
\end{figure}

Fig. \ref{experiment4_3} shows the space and time consumption of the proposed method and GDS method to get the deployment strategy under different servers. With the increase in the number of servers and microservices in the network, the amount of layers and the communication between microservices is also increasing. Then the network structure becomes more and more complex. When the amount of servers is less than 10, the computing time of the proposed method is less than 20 seconds and the space occupation is less than 100 MB, which has excellent solution efficiency. In a large-scale server network, the solution time will slow down to about 450 seconds, and the space occupation will also increase to 1800 MB due to the complexity of the network. The optimal solution can still be solved in an acceptable time because the scheduling strategy of microservices does not change frequently in large-scale production. However, since the proposed method is based on solving quadratic programming problems, the time and space complexity is higher than that of Algorithm \ref{alg:Simple}. It is our future work to further optimize the time complexity and space complexity of the proposed method.


\section{Discussion} \label{discussion}
 
In this section, we will discuss the limitations of the proposed method and the future work.

\begin{figure}[t]
  \centering
  \includegraphics[width=0.4\textwidth]{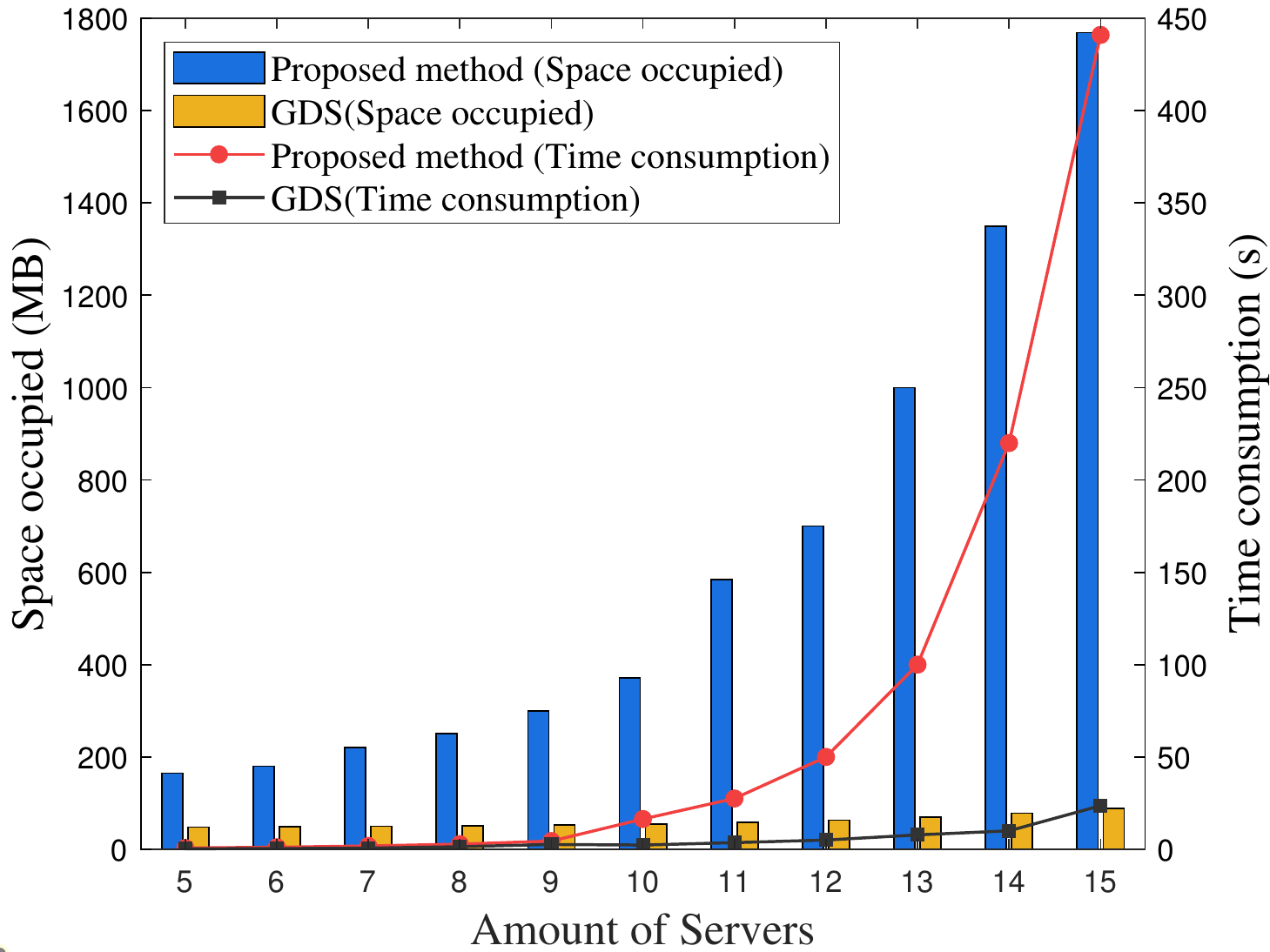}
  \caption{Space and time consumption under different amount of servers}
  \label{experiment4_3}
\end{figure}

Compared with related works, our proposed method is able to optimize the communication overhead while optimizing the image pull delay. The experiments in Section. \ref{performance} also show that the proposed method can achieve the lowest objective function value with a good trade-off between delay and overhead. However, the time and space complexity of this method is high in large-scale scenarios. Although the optimal solution can be obtained within an acceptable time, further optimization is required in the future.

The communication overhead depends not only on the amount of communication data, but also on the request frequency. If the request frequency is high, the communication overhead will be very large even if the amount of communication data at one time is small. Therefore, we believe that the calculation of communication overhead should consider the request frequency, that is, to consider the product of the communication data volume and the request frequency. However, since our method is based on microservices not requests, we cannot directly get the request frequency. Therefore, we believe that future work can consider the request frequency as an input, and then use the product of the request frequency and the amount of communication data as the calculation of communication overhead.

Deployed microservices are not static, and the entire production process will include shutdown, migration, and startup of new microservices. In the face of the dynamic microservice deployment process, it is necessary to have corresponding deployment algorithms to adapt to dynamic scenarios. One possible future research direction is to use artificial intelligence or other methods for training after the initial deployment results such that the microservices can be dynamically adjusted.

\section{Conclusion} \label{conclusion}

In this paper, we study the layer sharing and chain sharing of microservices and explore a microservice deployment scheme that can balance the two ways of resource sharing. We build an image pull delay and communication overhead minimization problem. We transform the problem into a linearly constrained integer quadratic programming problem through model reconstruction and obtain the deployment strategy through a successive convex approximation (SCA) method. Further, we propose a resource reallocation algorithm to fully utilize the idle resources of the server. Experimental results show that the proposed deployment strategy can balance the two resource sharing methods of microservices. When considering the two sharing methods in balance, the average image pull delay can be reduced to 65\% of the baseline, and the average communication overhead can be reduced to 30\% of the baseline. In the future, we will expand microservice deployment from static scenarios to highly dynamic scenarios and try to obtain rapid solution algorithms in large-scale scenarios.

\section*{Acknowledgment}

This work was supported by the National Key Research and Development Program of China (Grant No.2018YFB1702300), and in part by the NSF of China (Grants No. 61731012, 62025305, 61933009, and 92167205).

\appendices
\section{Matrix value} \label{matrix}

\begin{enumerate}
  \item 
  \begin{equation}
    \mathcal{Q}=\begin{bmatrix}\mathbf{Q}^1&\cdots&0\\\vdots&\ddots&\vdots\\0&\cdots&\mathbf{Q}^K\\\end{bmatrix}_{K\times K}
  \end{equation}
  where 
  \begin{align}
  \mathbf{Q}^\mathbf{k}&=\begin{bmatrix}\mathbf{q}&\cdots&0\\\vdots&\ddots&\vdots\\0&\cdots&\mathbf{q}\\\end{bmatrix}_{A_k\times A_k} \nonumber\\
  \mathbf{q}&= [1,1,\cdots,1]_{1\times N} \nonumber
  \end{align}
  \item
  \begin{equation}
    \mathbf{b}_1=[1,1,\cdots,1]_{1\times\sum_{k\in\mathbb{K}} A_k}^T
  \end{equation}
  \item 
  \begin{equation}
    \mathbf{H} = \begin{bmatrix}\mathbf{H}^1  &\cdots & 0 \\\vdots & \ddots & \vdots \\0 & \cdots & \mathbf{H}^K \\\end{bmatrix}_{K \times K}
  \end{equation}
  where 
  \begin{equation}
    \mathbf{H}^k = [\mathbf{I}_{n \times n} \quad \mathbf{0} \quad \cdots \quad \mathbf{0} ]_{1 \times A_k}\nonumber
  \end{equation}
  \item 
  \begin{equation}
    \mathbf{b}_2 = [\mathbf{x}^{1,0}, \cdots, \mathbf{x}^{K0}]^T
  \end{equation}
  where $\mathbf{x}^{k0} = [0,\cdots,1,\cdots,0]^T$, $x^{k0}_{N^k} = 1$
  \item 
  \begin{equation}
    \mathbf{Y}=\left[\mathbf{Y}_1,\cdots,\mathbf{Y}_N\right]^T
  \end{equation}
  where
  \begin{align}
    \mathbf{Y}_n&=\begin{bmatrix}\mathbf{P}_n^1\mathbf{V}_1^1\mathbf{E}^1&\cdots&\mathbf{P}_n^1\mathbf{V}_L^1\mathbf{E}^1\\\vdots&\ddots&\vdots\\\mathbf{P}_n^K\mathbf{V}_1^K\mathbf{E}^K&\cdots&\mathbf{P}_n^K\mathbf{V}_L^K\mathbf{E}^K\\\end{bmatrix}_{K\times L} \nonumber\\
    \mathbf{P}_n^k&=\begin{bmatrix}\mathbf{p}_n^N&\cdots&0\\\vdots&\ddots&\vdots\\0&\cdots&\mathbf{p}_n^N\\\end{bmatrix}_{A_k\times A_k} \nonumber \\
    \mathbf{p}_n^N &=[0,\cdots,0,1,0,\cdots,0]_{1\times N}^T , \mathbf{p}_n^N(n) = 1 \nonumber\\
    \mathbf{q}&= [1,1,\cdots,1]_{1\times N} \nonumber \\
    \mathbf{V}_l^k&=\begin{bmatrix}{{(\mathbf{p}}_l^L)}^T&\cdots&0\\\vdots&\ddots&\vdots\\0&\cdots&{{(\mathbf{p}}_l^L)}^T\\\end{bmatrix}_{A_k\times A_k}\nonumber \\
    \mathbf{p}_l^L&= [0,\cdots,0,1,0,\cdots,0]_{1\times L}^T, \mathbf{p}_l^L(l) = 1 \nonumber\\
    \mathbf{E}^k&=\left[\left(\mathbf{E}^{k1}\right)^T,\cdots,\left(\mathbf{E}^{kA_k}\right)^T\right]^T \nonumber\\
    \mathbf{E}^{ki}&=\left[E^{ki1},\cdots,E^{kiL}\right]^T \nonumber
  \end{align}
  \item 
  \begin{equation}
  \mathcal{S}=\begin{bmatrix}\mathbf{S}^T&\cdots&0\\\vdots&\ddots&\vdots\\0&\cdots&\mathbf{S}^T\\\end{bmatrix}_{N\times N}
  \end{equation}
  where 
  \begin{equation}
  \mathbf{S}=\left[S^1,\cdots,S^L\right]^T \nonumber 
  \end{equation}
  \item 
  \begin{equation}
  \mathbf{C}^S=\left[C_1^S,\cdots,C_N^S\right]^T
  \end{equation}
  \item 
  \begin{equation}
  \mathcal{G}=\left[\mathbf{G}_1,\cdots,\mathbf{G}_N\right]^T
  \end{equation}
  where 
  \begin{align}
    \mathbf{G}_n&=\left[\left(\mathbf{P}_n^1\mathbf{u}^1\right)^T,\cdots,\left(\mathbf{P}_n^K\mathbf{u}^K\right)^T\right]^T \nonumber \\
    \mathbf{u}^k&=\left[u^{k1},\cdots,u^{kA_k}\right]^T \nonumber 
  \end{align}
  \item 
  \begin{equation}
    \mathbf{C}^C = \left[C_1^C, \cdots, C_N^C\right]^T
  \end{equation}
\end{enumerate}

\section{Proof of Theorem 1} \label{proof1}

\begin{proof}
We remove the bolding of all letters in this proof to simplify the expression. We combine the same parts of the constraints in $\text{P3}$ and rewrite it as

\begin{align}
  \min F(x,d) &= C_1Md + \frac{1}{2}C_2x^TPx \nonumber \\ &\qquad - C_2\bar{x}^TNx + \frac{1}{2}C_2\bar{x}^TN\bar{x} \nonumber \\
  \text{s.t.} \quad& a_i^Tx= b_i, i=1, \cdots, m \nonumber \\
  & a_i^Tx \leqslant b_i, i=m, \cdots, m+l \nonumber \\
  & c_i^Ty \leqslant g_i, i=1, \cdots, n \nonumber \\
  & e_ix + f_iy \leqslant 0, i=1, \cdots, p \nonumber
\end{align}

Its KKT condition is

\begin{equation}
  \begin{cases}
  C_2P\bar{x} - C_2N\bar{x} + \sum_{i=1}^{m+l}\lambda_ia_i + \sum_{i=1}^{k}\nu e_i = 0 \\
  C_1M + \sum_{i=1}^{n}\mu_ic_i + \sum_{i=1}^{k}\nu e_i = 0 \\
  a_i^Tx= b_i, i=1, \cdots, m \\
  a_i^Tx \leqslant b_i, i=m, \cdots, m+l \\
  c_i^Ty \leqslant  g_i, i=1, \cdots, n \\
  e_ix + f_iy \leqslant 0, i=1, \cdots, p\\
  \lambda_i \geqslant 0 , i=1, \cdots, m + l\\
  \mu_i \geqslant 0 , i=1, \cdots, n\\
  \nu_i \geqslant 0 , i=1, \cdots, p\\
  \lambda_i(a_i^Tx - b_i) = 0, i=m, \cdots, m+l \\
  \mu_i(c_i^Ty - g_i) = 0, i=1, \cdots, n \\
  \nu_i(e_ix + f_iy) = 0, i=1, \cdots, p
  \end{cases}
\end{equation}

Since $P\bar{x} - N\bar{x} = Q\bar{x}$, so

\begin{equation}
  \begin{cases}
  C_2Q\bar{x} + \sum_{i=1}^{m+l}\lambda_ia_i + \sum_{i=1}^{k}\nu e_i = 0 \\
  C_1M + \sum_{i=1}^{n}\mu_ic_i + \sum_{i=1}^{k}\nu e_i = 0 \\
  a_i^Tx= b_i, i=1, \cdots, m \\
  a_i^Tx \leqslant b_i, i=m, \cdots, m+l \\
  c_i^Ty \leqslant  g_i, i=1, \cdots, n \\
  e_ix + f_iy \leqslant 0, i=1, \cdots, p\\
  \lambda_i \geqslant 0 , i=1, \cdots, m + l\\
  \mu_i \geqslant 0 , i=1, \cdots, n\\
  \nu_i \geqslant 0 , i=1, \cdots, p\\
  \lambda_i(a_i^Tx - b_i) = 0, i=m, \cdots, m+l \\
  \mu_i(c_i^Ty - g_i) = 0, i=1, \cdots, n \\
  \nu_i(e_ix + f_iy) = 0, i=1, \cdots, p
  \end{cases}
\end{equation}

Therefore, $\bar{x}$ is the KKT point of $\text{P3}$, the non-global solution of the original problem can be obtained.
\end{proof}

\section{Proof of Theorem 2} \label{proof2}

\begin{proof}
We remove the bolding of all letters in this proof to simplify the expression. We can find that the search direction of the algorithm is $\lambda_x^{r+1} = z_x^{r+1} - x^{r},\lambda_y^{r+1} = z_d^{r+1} - d^{r}$. According to the convexity of $\text{P4}$, we can get

\begin{equation}
\begin{array}{l}
  U(x^{r+1},d^{r+1}) \leqslant U_{qp}(x^{r+1},d^{r+1};x^{r},d^{r}) \\
  \leqslant \alpha U_{qp}(x^{r+1},z_d^{r+1};x^{r},d^{r}) \\ \quad + (1-\alpha)U_{qp}(x^{r+1},d^r;x^{r},d^{r}) \\
  \leqslant U_{qp}(x^{r+1},d^r;x^{r},d^{r}) \\
  \leqslant \alpha U_{qp}(z_x^{r+1},d^r;x^{r},d^{r}) + (1-\alpha)U_{qp}(x^r,d^r;x^r,d^r) \\
  \leqslant U_{qp}(x^r,d^r;x^r,d^r) \\
  = U(x^{r},d^{r})
\end{array}
\end{equation}
where the second and the fourth inequality sign come from the convexity of $U_{qp}(\cdot;x^r,y^r)$. Then the value of the objective function must not be monotone increasing.

\begin{equation}
\begin{array}{l}
  U(x^r,d^r) - U(x^{r+1},d^{r+1}) \\
  = U(x^r,d^r) - U(x^{r+1},d^{r}) \\ \quad + U(x^{r+1},d^{r}) - U(x^{r+1},d^{r+1}) \\
  = U(x^r,d^r) - U(x^{r} + \alpha(z_x^{r+1} - x^{r}),d^{r}) \\ \quad + U(x^{r+1},d^{r}) - U(x^{r+1} ,d^{r} + \alpha(z_d^{r+1} - d^{r})) \\
  \geqslant -\alpha U'(x^r,d^r;\lambda_x^{r+1}) - \alpha U'(x^{r+1},d^r;\lambda_d^{r+1}) \\
  \geqslant 0
\end{array}
\end{equation}

Then we can get 

\begin{equation}
\begin{array}{l}
  U(x^*,d^*) - U(x^0,d^0) \\
  \leqslant U(x^r,d^r)  - U(x^0,d^0) \\
  \leqslant \sum_{r=0}^r\alpha(U'(x^r,d^r;\lambda_x^{r+1}) +U'(x^{r+1},d^r;\lambda_d^{r+1}))
\end{array}
\end{equation}

thereby

\begin{equation}
\lim_{r \rightarrow \infty }U'(x^r,d^r;\lambda_x^{r+1}) = \lim_{k \rightarrow \infty}U'(x^r,d^r;\lambda_y^{r+1})  = 0
\end{equation}

So problem $\text{P3}$ can get a stationary solution by SCA algorithm. We can choose $\alpha=1$ for integer variables, and it still holds.
\end{proof}

\ifCLASSOPTIONcaptionsoff
  \newpage
\fi



%
\bibliographystyle{IEEEtran}
\bibliography{mylib}

\begin{thebibliography}{10}
\providecommand{\url}[1]{#1}
\csname url@samestyle\endcsname
\providecommand{\newblock}{\relax}
\providecommand{\bibinfo}[2]{#2}
\providecommand{\BIBentrySTDinterwordspacing}{\spaceskip=0pt\relax}
\providecommand{\BIBentryALTinterwordstretchfactor}{4}
\providecommand{\BIBentryALTinterwordspacing}{\spaceskip=\fontdimen2\font plus
\BIBentryALTinterwordstretchfactor\fontdimen3\font minus
  \fontdimen4\font\relax}
\providecommand{\BIBforeignlanguage}[2]{{%
\expandafter\ifx\csname l@#1\endcsname\relax
\typeout{** WARNING: IEEEtran.bst: No hyphenation pattern has been}%
\typeout{** loaded for the language `#1'. Using the pattern for}%
\typeout{** the default language instead.}%
\else
\language=\csname l@#1\endcsname
\fi
#2}}
\providecommand{\BIBdecl}{\relax}
\BIBdecl

\bibitem{wangLowlatencyInteroperableIndustrial2020}
R.~Wang, L.~Ji, T.~Ren, S.~He, and Z.~Shi, ``A {{Low-latency}} and
  {{Interoperable Industrial Internet}} of {{Things Architecture}} for
  {{Manufacturing Systems}},'' in \emph{2020 {{IEEE}} 18th {{International
  Conference}} on {{Industrial Informatics}} ({{INDIN}})}.\hskip 1em plus 0.5em
  minus 0.4em\relax {Warwick, United Kingdom}: {IEEE}, Jul. 2020, pp. 859--864.

\bibitem{chenResourceRecommendationModel2022}
L.~Chen, Z.~Lu, A.~Xiao, Q.~Duan, J.~Wu, and P.~C.~K. Hung, ``A {{Resource
  Recommendation Model}} for {{Heterogeneous Workloads}} in {{Fog-Based Smart
  Factory Environment}},'' \emph{IEEE Transactions on Automation Science and
  Engineering}, vol.~19, no.~3, pp. 1731--1743, Jul. 2022.

\bibitem{mendoncaDevelopingSelfAdaptiveMicroservice2021}
N.~C. Mendonca, P.~Jamshidi, D.~Garlan, and C.~Pahl, ``Developing
  {{Self-Adaptive Microservice Systems}}: {{Challenges}} and {{Directions}},''
  \emph{IEEE Software}, vol.~38, no.~2, pp. 70--79, Mar. 2021.

\bibitem{nieHypergraphicalRealtimeMultiRobot2021}
Z.~Nie and K.-C. Chen, ``Hypergraphical {{Real-time Multi-Robot Task
  Allocation}} in a {{Smart Factory}},'' \emph{IEEE Transactions on Industrial
  Informatics}, pp. 1--1, 2021.

\bibitem{hazraCooperativeTransmissionScheduling2022}
A.~Hazra, P.~K. Donta, T.~Amgoth, and S.~Dustdar, ``Cooperative {{Transmission
  Scheduling}} and {{Computation Offloading}} with {{Collaboration}} of {{Fog}}
  and {{Cloud}} for {{Industrial IoT Applications}},'' \emph{IEEE Internet of
  Things Journal}, pp. 1--1, 2022.

\bibitem{thramboulidisCyberphysicalMicroservicesIoTbased2018}
K.~Thramboulidis, D.~C. Vachtsevanou, and A.~Solanos, ``Cyber-physical
  microservices: {{An IoT-based}} framework for manufacturing systems,'' in
  \emph{2018 {{IEEE Industrial Cyber-Physical Systems}} ({{ICPS}})}, 2018, pp.
  232--239.

\bibitem{yangMIRASModelbasedReinforcement2019}
Z.~Yang, P.~Nguyen, H.~Jin, and K.~Nahrstedt, ``{{MIRAS}}: {{Model-based
  Reinforcement Learning}} for {{Microservice Resource Allocation}} over
  {{Scientific Workflows}},'' in \emph{2019 {{IEEE}} 39th {{International
  Conference}} on {{Distributed Computing Systems}} ({{ICDCS}})}, Jul. 2019,
  pp. 122--132.

\bibitem{tianDIMADistributedCooperative2021}
H.~Tian, X.~Xu, T.~Lin, Y.~Cheng, C.~Qian, L.~Ren, and M.~Bilal, ``{{DIMA}}:
  {{Distributed}} cooperative microservice caching for internet of things in
  edge computing by deep reinforcement learning,'' \emph{World Wide Web}, Aug.
  2021.

\bibitem{baoPerformanceModelingWorkflow2019}
L.~Bao, C.~Wu, X.~Bu, N.~Ren, and M.~Shen, ``Performance {{Modeling}} and
  {{Workflow Scheduling}} of {{Microservice-Based Applications}} in
  {{Clouds}},'' \emph{IEEE Transactions on Parallel and Distributed Systems},
  vol.~30, no.~9, pp. 2114--2129, Sep. 2019.

\bibitem{jianCloudEdgebasedTwolevel2021}
C.~Jian, J.~Ping, and M.~Zhang, ``A cloud edge-based two-level hybrid
  scheduling learning model in cloud manufacturing,'' \emph{International
  Journal of Production Research}, vol.~59, no.~16, pp. 4836--4850, Aug. 2021.

\bibitem{shiEdgeComputingVision2016}
W.~Shi, J.~Cao, Q.~Zhang, Y.~Li, and L.~Xu, ``Edge {{Computing}}: {{Vision}}
  and {{Challenges}},'' \emph{IEEE Internet of Things Journal}, vol.~3, no.~5,
  pp. 637--646, Oct. 2016.

\bibitem{zhangFengHuoLunFederatedLearning2020}
C.~Zhang, X.~Liu, X.~Zheng, R.~Li, and H.~Liu, ``{{FengHuoLun}}: {{A Federated
  Learning}} based {{Edge Computing Platform}} for {{Cyber-Physical
  Systems}},'' in \emph{2020 {{IEEE International Conference}} on {{Pervasive
  Computing}} and {{Communications Workshops}} ({{PerCom Workshops}})}, Mar.
  2020, pp. 1--4.

\bibitem{merkelDockerLightweightLinux2014}
D.~Merkel, ``Docker: Lightweight {{Linux}} containers for consistent
  development and deployment,'' \emph{Linux Journal}, vol. 2014, no. 239, p.
  2:2, Mar. 2014.

\bibitem{burnsBorgOmegaKubernetes2016}
B.~Burns, B.~Grant, D.~Oppenheimer, E.~Brewer, and J.~Wilkes, ``Borg,
  {{Omega}}, and {{Kubernetes}},'' \emph{Communications of the ACM}, vol.~59,
  no.~5, pp. 50--57, Apr. 2016.

\bibitem{guExploringLayeredContainer2021}
L.~Gu, D.~Zeng, J.~Hu, H.~Jin, S.~Guo, and A.~Y. Zomaya, ``Exploring {{Layered
  Container Structure}} for {{Cost Efficient Microservice Deployment}},'' in
  \emph{{{IEEE INFOCOM}} 2021 - {{IEEE Conference}} on {{Computer
  Communications}}}.\hskip 1em plus 0.5em minus 0.4em\relax {Vancouver, BC,
  Canada}: {IEEE}, May 2021, pp. 1--9.

\bibitem{guNDockerNVMHDDHybrid2019}
L.~Gu, Q.~Tang, S.~Wu, H.~Jin, Y.~Zhang, G.~Shi, T.~Lin, and J.~Rao,
  ``N-{{Docker}}: {{A NVM-HDD Hybrid Docker Storage Framework}} to {{Improve
  Docker Performance}},'' in \emph{Network and {{Parallel Computing}}}, ser.
  Lecture {{Notes}} in {{Computer Science}}, X.~Tang, Q.~Chen, P.~Bose,
  W.~Zheng, and J.-L. Gaudiot, Eds.\hskip 1em plus 0.5em minus 0.4em\relax
  {Cham}: {Springer International Publishing}, 2019, pp. 182--194.

\bibitem{wangMPCSMMicroservicePlacement2021}
Y.~Wang, C.~Zhao, S.~Yang, X.~Ren, L.~Wang, P.~Zhao, and X.~Yang, ``{{MPCSM}}:
  {{Microservice Placement}} for {{Edge-Cloud Collaborative Smart
  Manufacturing}},'' \emph{IEEE Transactions on Industrial Informatics},
  vol.~17, no.~9, pp. 5898--5908, Sep. 2021.

\bibitem{lvMicroserviceDeploymentEdge2022}
W.~Lv, Q.~Wang, P.~Yang, Y.~Ding, B.~Yi, Z.~Wang, and C.~Lin, ``Microservice
  {{Deployment}} in {{Edge Computing Based}} on {{Deep Q Learning}},''
  \emph{IEEE Transactions on Parallel and Distributed Systems}, pp. 1--1, 2022.

\bibitem{yuJointOptimizationService2019}
Y.~Yu, J.~Yang, C.~Guo, H.~Zheng, and J.~He, ``Joint optimization of service
  request routing and instance placement in the microservice system,''
  \emph{Journal of Network and Computer Applications}, vol. 147, p. 102441,
  Dec. 2019.

\bibitem{herreraOptimizingResponseTime2021a}
J.~L. Herrera, J.~{Gal{\'a}n-Jim{\'e}nez}, J.~Berrocal, and J.~M. Murillo,
  ``Optimizing the {{Response Time}} in {{SDN-Fog Environments}} for
  {{Time-Strict IoT Applications}},'' \emph{IEEE Internet of Things Journal},
  vol.~8, no.~23, pp. 17\,172--17\,185, Dec. 2021.

\bibitem{dengOptimalApplicationDeployment2021}
S.~Deng, Z.~Xiang, J.~Taheri, M.~A. Khoshkholghi, J.~Yin, A.~Y. Zomaya, and
  S.~Dustdar, ``Optimal {{Application Deployment}} in {{Resource Constrained
  Distributed Edges}},'' \emph{IEEE Transactions on Mobile Computing}, vol.~20,
  no.~5, pp. 1907--1923, May 2021.

\bibitem{chenDynamicServiceMigration2022}
X.~Chen, Y.~Bi, X.~Chen, H.~Zhao, N.~Cheng, F.~Li, and W.~Cheng, ``Dynamic
  {{Service Migration}} and {{Request Routing}} for {{Microservice}} in
  {{Multi-cell Mobile Edge Computing}},'' \emph{IEEE Internet of Things
  Journal}, pp. 1--1, 2022.

\bibitem{faddaMonitoringAwareOptimalDeployment2021}
E.~Fadda, P.~Plebani, and M.~Vitali, ``Monitoring-{{Aware Optimal Deployment}}
  for {{Applications Based}} on {{Microservices}},'' \emph{IEEE Transactions on
  Services Computing}, vol.~14, no.~6, pp. 1849--1863, 2021.

\bibitem{zhaoCostEfficientEdgeIntelligent2020a}
P.~Zhao, P.~Wang, X.~Yang, and J.~Lin, ``Towards {{Cost-Efficient Edge
  Intelligent Computing With Elastic Deployment}} of {{Container-Based
  Microservices}},'' \emph{IEEE Access}, vol.~8, pp. 102\,947--102\,957, 2020.

\bibitem{faticantiThroughputAwarePartitioningPlacement2020}
F.~Faticanti, F.~De~Pellegrini, D.~Siracusa, D.~Santoro, and S.~Cretti,
  ``Throughput-{{Aware Partitioning}} and {{Placement}} of {{Applications}} in
  {{Fog Computing}},'' \emph{IEEE Transactions on Network and Service
  Management}, vol.~17, no.~4, pp. 2436--2450, Dec. 2020.

\bibitem{josephIntMADynamicInteractionaware2020}
C.~T. Joseph and K.~Chandrasekaran, ``{{IntMA}}: {{Dynamic Interaction-aware}}
  resource allocation for containerized microservices in cloud environments,''
  \emph{Journal of Systems Architecture}, vol. 111, p. 101785, Dec. 2020.

\bibitem{liOnlineReconfigurationLatencyAware2021}
X.~Li, Z.~Zhou, C.~Zhu, L.~Shu, and J.~Zhou, ``Online {{Reconfiguration}} of
  {{Latency-Aware IoT Services}} in {{Edge Networks}},'' \emph{IEEE Internet of
  Things Journal}, pp. 1--12, 2021.

\bibitem{armaniCostEffectiveWorkloadAllocation2021}
V.~Armani, F.~Faticanti, S.~Cretti, S.~Kum, and D.~Siracusa, ``A
  {{Cost-Effective Workload Allocation Strategy}} for {{Cloud-Native Edge
  Services}},'' \emph{arXiv:2110.12788 [cs]}, Oct. 2021.

\bibitem{sasabeCapacitatedShortestPath2021}
M.~Sasabe and T.~Hara, ``Capacitated {{Shortest Path Tour Problem-Based Integer
  Linear Programming}} for {{Service Chaining}} and {{Function Placement}} in
  {{NFV Networks}},'' \emph{IEEE Transactions on Network and Service
  Management}, vol.~18, no.~1, pp. 104--117, Mar. 2021.

\bibitem{louEfficientContainerAssignment2022}
J.~Lou, H.~Luo, Z.~Tang, W.~Jia, and W.~Zhao, ``Efficient {{Container
  Assignment}} and {{Layer Sequencing}} in {{Edge Computing}},'' \emph{IEEE
  Transactions on Services Computing}, pp. 1--1, 2022.

\bibitem{guLayerAwareMicroservice2021}
L.~Gu, D.~Zeng, J.~Hu, B.~Li, and H.~Jin, ``Layer {{Aware Microservice
  Placement}} and {{Request Scheduling}} at the {{Edge}},'' in \emph{{{IEEE
  INFOCOM}} 2021 - {{IEEE Conference}} on {{Computer Communications}}}.\hskip
  1em plus 0.5em minus 0.4em\relax {Vancouver, BC, Canada}: {IEEE}, May 2021,
  pp. 1--9.

\bibitem{guLayerawareCollaborativeMicroservice2022}
L.~Gu, Z.~Chen, H.~Xu, D.~Zeng, B.~Li, and H.~Jin, ``Layer-aware
  {{Collaborative Microservice Deployment}} toward {{Maximal Edge
  Throughput}},'' in \emph{{{IEEE INFOCOM}} 2022 - {{IEEE Conference}} on
  {{Computer Communications}}}, 2022, pp. 71--79.

\bibitem{wangDelayAwareMicroserviceCoordination2021}
S.~Wang, Y.~Guo, N.~Zhang, P.~Yang, A.~Zhou, and X.~Shen, ``Delay-{{Aware
  Microservice Coordination}} in {{Mobile Edge Computing}}: {{A Reinforcement
  Learning Approach}},'' \emph{IEEE Transactions on Mobile Computing}, vol.~20,
  no.~3, pp. 939--951, Mar. 2021.

\bibitem{razaviyaynSuccessiveConvexApproximation2014}
M.~Razaviyayn, ``Successive convex approximation: Analysis and applications,''
  \emph{University of Minnesota}, May 2014.

\bibitem{GurobiFastestSolver}
``Gurobi - {{The Fastest Solver}},'' https://www.gurobi.com/.

\bibitem{fuFastEfficientContainer2020}
S.~Fu, R.~Mittal, L.~Zhang, and S.~Ratnasamy, ``Fast and {{Efficient Container
  Startup}} at the {{Edge}} via {{Dependency Scheduling}},'' \emph{3rd USENIX
  Workshop on Hot Topics in Edge Computing (HotEdge 20)}, p.~7, 2020.

\bibitem{KubernetesKubernetes}
``Kubernetes, {{Kubernetes}}.'' https://kubernetes.io/.

\end{thebibliography}

%








\begin{IEEEbiography}[{\includegraphics[width=1in,height=1.25in,clip,keepaspectratio]{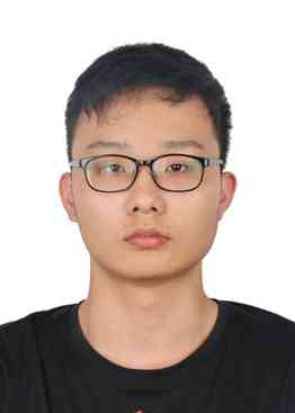}}]{Yuxiang Liu}
	received the B.Eng. degree in control science and engineering from Shanghai Jiao Tong University, Shanghao, China, in 2020. 
	
	He is currently pursuing the Ph.D. degree at the Department of Automation, Shanghai Jiao Tong University, Shanghai, China. His current research interests include microservice deployment, intelligent manufacturing.
\end{IEEEbiography}

\begin{IEEEbiography}[{\includegraphics[width=1in,height=1.25in,clip,keepaspectratio]{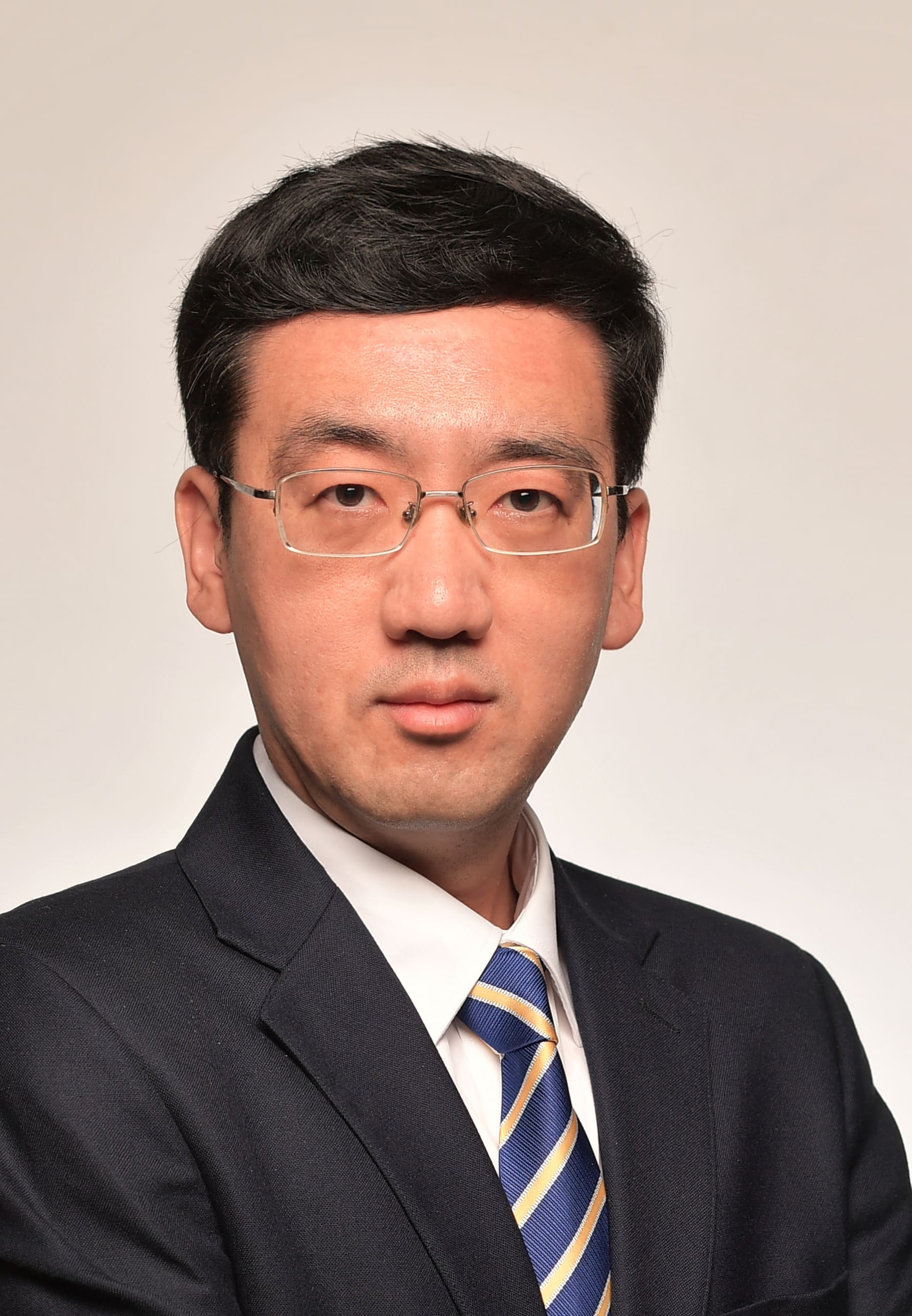}}]{Bo Yang}
	(Senior Member, IEEE) received the Ph.D. degree in electrical engineering from the City University of Hong Kong, Hong Kong, in 2009. 

	He is currently a Full Professor with Shanghai Jiao Tong University, Shanghai, China. Prior to joining Shanghai Jiao Tong University in 2010, he was a Post-Doctoral Researcher with the KTH Royal Institute of Technology, Stockholm, Sweden, from 2009 to 2010, and a Visiting Scholar with the Polytechnic Institute of New York University in 2007. 
	His research interests include game theoretical analysis and optimization of energy networks and wireless networks.
	He is on the Editorial Board of Digital Signal Processing and in TPC of several international conferences. He has been the Principle Investigator in several research projects, including the NSFC Key Project. He was a recipient of the Ministry of Education Natural Science Award 2016, the Shanghai Technological Invention Award 2017, the Shanghai Rising-Star Program 2015, and the SMC-Excellent Young Faculty Award by Shanghai Jiao Tong University.
\end{IEEEbiography}

\begin{IEEEbiography}[{\includegraphics[width=1in,height=1.25in,clip,keepaspectratio]{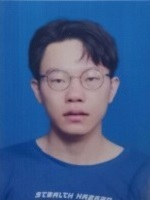}}]{Yu Wu}
	received the B.Eng. degree at the Department of Automation, Harbin Engineering University, Harbin, China, in 2019. He is currently pursuing the Ph.D. degree at the Department of Automation, Shanghai Jiao Tong University, Shanghai, China. His current research interests include edge computing, industrial Internet of Things, and machine learning for wireless networks.
\end{IEEEbiography}

\begin{IEEEbiography}[{\includegraphics[width=1in,height=1.25in,clip,keepaspectratio]{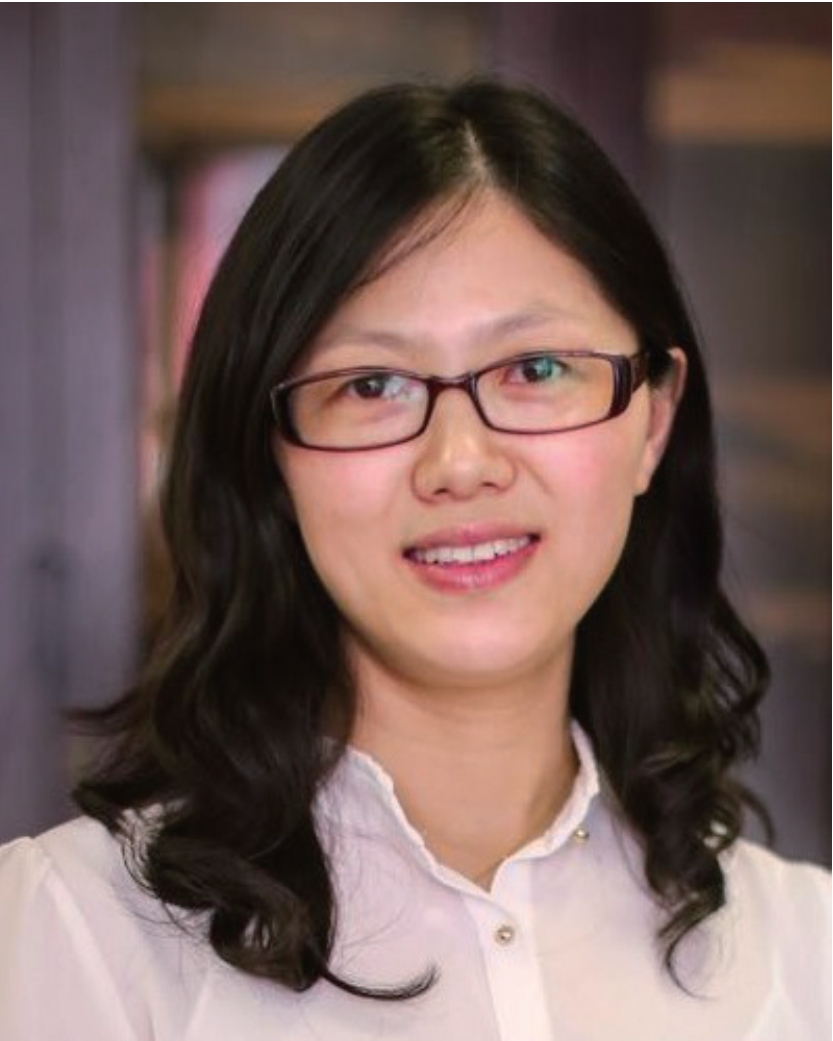}}]{Cailian Chen} 
	Cailian Chen (S'03-M'06) received the B. Eng. and M. Eng. degrees in Automatic Control from Yanshan University, P. R. China in 2000 and 2002, respectively, and the Ph.D. degree in Control and Systems from City University of Hong Kong, Hong Kong SAR in 2006. She has been with the Department of Automation, Shanghai Jiao Tong University since 2008. She is now a Distinguished Professor. 

	Prof. Chen’s research interests include industrial wireless networks and computational intelligence, and Internet of Vehicles. She has authored 3 research monographs and over 100 referred international journal papers. She is the inventor of more than 30 patents. Dr. Chen received the prestigious "IEEE Transactions on Fuzzy Systems Outstanding Paper Award" in 2008, and 5 conference best paper awards. She won the Second Prize of National Natural Science Award from the State Council of China in 2018, First Prize of Natural Science Award from The Ministry of Education of China in 2006 and 2016, respectively, and First Prize of Technological Invention of Shanghai Municipal, China in 2017. She was honored “National Outstanding Young Researcher” by NSF of China in 2020 and “Changjiang Young Scholar” in 2015.
\end{IEEEbiography}

\begin{IEEEbiography}[{\includegraphics[width=1in,height=1.25in,clip,keepaspectratio]{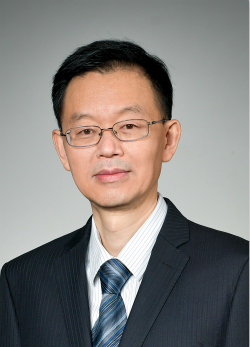}}]{Xinping Guan} 
	(Fellow, IEEE) received the B.Sc. degree in mathematics from Harbin Normal University, Harbin, China, in 1986, and the Ph.D. degree in control science and engineering from Harbin Institute of Technology, Harbin, China, in 1999.

	He is currently the Chair Professor of Shanghai Jiao Tong University, Shanghai, China, where he is the Dean of School of Electronic, Information and Electrical Engineering and the Director of the Key Laboratory of Systems Control and Information Processing, Ministry of Education of China. Before that, he was the Executive Director of Office of Research Management, Shanghai Jiao Tong University, a Full Professor, and Dean of Electrical Engineering, Yanshan University, Qinhuangdao, China. As a Principal Investigator, he has finished/been working on more than 20 national key projects. He is the leader of the prestigious Innovative Research Team of the National Natural Science Foundation of China. He is an Executive Committee Member of Chinese Automation Association Council and the Chinese Artificial Intelligence Association Council. He has authored or coauthored five research monographs, more than 200 papers in IEEE transactions and other peer-reviewed journals, and numerous conference papers. His current research interests include industrial network systems, smart manufacturing, and underwater networks.

	Dr. Guan received the Second Prize of the National Natural Science Award of China in both 2008 and 2018, the First Prize of Natural Science Award from the Ministry of Education of China in both 2006 and 2016. He was a recipient of the IEEE Transactions on Fuzzy Systems Outstanding Paper Award in 2008. He is a National Outstanding Youth honored by NSF of China, and Changjiang Scholar's by the Ministry of Education of China and State-Level Scholar of New Century Bai Qianwan Talent Program of China.
\end{IEEEbiography}

\end{document}